\documentclass[12pt,onecolumn]{IEEEtran}

\usepackage[T1]{fontenc}
\usepackage{amsmath,amssymb,amsthm,mathtools}
\usepackage{graphicx}
\usepackage{float}
\usepackage{cite}
\usepackage{microtype}
\usepackage{xcolor}
\usepackage[hidelinks]{hyperref}

\hypersetup{
  pdftitle={Second-Order Expansion of Privacy Amplification Under f-Divergence Criteria},
  pdfauthor={Mario Berta, Hao-Chung Cheng, and Marco Tomamichel},
  colorlinks=true,
  linkcolor=black,
  citecolor=black,
  urlcolor=blue!55!black
}

\newtheorem{theorem}{Theorem}
\newtheorem{proposition}[theorem]{Proposition}
\newtheorem{lemma}[theorem]{Lemma}
\newtheorem{corollary}[theorem]{Corollary}
\newtheorem{fact}[theorem]{Fact}
\theoremstyle{definition}
\newtheorem{definition}[theorem]{Definition}
\theoremstyle{remark}

\DeclareMathOperator{\Var}{Var}
\DeclareMathOperator{\supp}{supp}
\newcommand{\E}{\mathbb{E}}
\newcommand{\Prb}{\mathbb{P}}
\newcommand{\ind}{\mathbf{1}}
\newcommand{\cX}{\mathcal{X}}
\newcommand{\cY}{\mathcal{Y}}
\newcommand{\cZ}{\mathcal{Z}}
\newcommand{\cS}{\mathcal{S}}
\newcommand{\Fprof}{F}
\newcommand{\Gprof}{G}

\begin{document}

\title{Second-Order Expansion of Privacy Amplification
Under $f$-Divergence Criteria}

\author{Mario Berta, Hao-Chung Cheng, and Marco Tomamichel%
\thanks{M. Berta is with the Institute for Quantum Information, RWTH Aachen
University, Germany and the Department of Computing, Imperial College London, United Kingdom (e-mail: berta@physik.rwth-aachen.de).}%
\thanks{H.-C. Cheng is with the Department of Electrical Engineering and the
Graduate Institute of Communication Engineering, National Taiwan University,
Taipei 106, Taiwan (e-mail: haochung@ntu.edu.tw).}%
\thanks{M. Tomamichel is with the Department of Electrical and Computer
Engineering and the Centre for Quantum Technologies, National University of
Singapore, Singapore 117583 (e-mail: marco.tomamichel@nus.edu.sg).}}

\maketitle

\begin{abstract}
We derive the second-order asymptotics of randomness extraction from memoryless sources with side information under security criteria based on a broad class of Csisz\'ar $f$-divergences, treating both a fixed reference side-information marginal and optimization over that marginal. The conditional varentropy decomposes into fluctuations of the conditional entropy across different values of the side information and the average variance of the conditional surprisal for each value. Without marginal optimization, these contributions yield a Gaussian-mixture second-order profile. With marginal optimization, they combine into the total conditional varentropy, yielding a single Gaussian profile. As corollaries, we obtain second-order expansions for Rényi-entropy criteria of all orders $\alpha\in(0,1)$ and recover the known expansion for total variation distance.
\end{abstract}


\begingroup
\par\addvspace{0.5\baselineskip}\centering\small
\textbf{Use of Artificial Intelligence}\par\addvspace{0.5\baselineskip}
\endgroup
\begin{quotation}\small
OpenAI Codex with ChatGPT 5.6 Sol was instrumental throughout this project
and performed most of the work on the literature search, exploration and
checking of proof strategies, initial manuscript preparation, and development
and audit of the Lean formalization.  The authors provided significant
guidance during the revision of the manuscript.  They have reviewed the
resulting manuscript and formalization and take responsibility for their
contents and any remaining errors.
\par
\end{quotation}

\begingroup
\par\addvspace{0.5\baselineskip}\centering\small
\textbf{Formal Verification}\par\addvspace{0.5\baselineskip}
\endgroup
\begin{quotation}\small
Every numbered result and definition has a corresponding declaration in an
accompanying Lean~4.33.1 development using Mathlib~4.33.1
\cite{de_moura_ullrich_2021,mathlib_2020}.  The development is checked by the
Lean kernel and introduces no project-specific axioms.
Theorem~\ref{thm:main-f} and its corollaries are formalized conditional only
on the finite-valued Berry--Esseen statement in Fact~\ref{fact:BE}.  The
source code, dependency graphs, and statement-by-statement correspondence are
available in the accompanying GitHub repository and archived on Zenodo
\cite{berta_cheng_tomamichel_lean_2026}.
\par
\end{quotation}
\vspace{0.67ex}

\section{Introduction}

\subsection{Motivation and problem setting}

Privacy amplification is the cryptographic task of turning partially secret
data into a shorter, nearly uniform secret key.  Let Alice observe $X$ while
an adversary observes correlated side information $Y$.  A public random seed
$S$, independent of $(X,Y)$, selects a hash $\varphi_s:\cX\to\cZ$, where
$M:=|\cZ|$, and Alice outputs $Z=\varphi_S(X)$.  Because the seed is public,
the adversary's complete view is $(Y,S)$.  Security requires $Z$ to be uniform
and independent of this view, so the ideal law is
\begin{equation}
 U_Z\times P_Y\times P_S,
 \label{eq:ideal-law}
\end{equation}
where $U_Z$ denotes the uniform distribution on $\cZ$.
We call the deviation of the actual law $P_{ZYS}$ from this ideal law,
measured by the chosen security criterion, the leakage.

This formulation originates in the cryptographic literature.  Bennett,
Brassard, and Robert introduced privacy amplification by public discussion
for parties sharing randomness about which a computationally unbounded
eavesdropper has partial information
\cite{bennett_brassard_robert_1988}.  Their constructions used universal
hashing, introduced by Carter and Wegman \cite{carter_wegman_1979}.  In
complexity theory, Impagliazzo, Levin, and Luby established the result now
known as the leftover-hash lemma, showing that universal hashing gives strong
seeded extraction in statistical, or total-variation, distance
\cite{impagliazzo_levin_luby_1989}.  Bennett, Brassard, Cr\'epeau, and Maurer
subsequently developed a general privacy-amplification framework based on
conditional collision entropy and universal hashing
\cite{bennett_brassard_crepeau_maurer_1995}.  In parallel, the
information-theoretic problem usually called intrinsic randomness or
random-number generation was studied without side information by Vembu and
Verd\'u \cite{vembu_verdu_1995}.  For arbitrary sources, not assumed
stationary or ergodic, they identified the optimal fixed-length rate with the
inf-entropy rate.  Further developments for source models without side
information include
\cite{hayashi_2008,nomura_han_2013,uyematsu_matsuta_2017,hayashi_watanabe_2016}.
For approximation criteria closer to those considered here, Yu and Tan
studied random-variable simulation under R\'enyi divergences of all orders,
with unconditioned intrinsic randomness as the specialization in which the
target law is uniform \cite{yu_tan_2019}.  Nomura subsequently established
first-, second-, and optimistic-order information-spectrum formulas for
unconditioned intrinsic randomness under a subclass of $f$-divergences
\cite{nomura_2020}; see also \cite{nomura_yagi_2024}.  When side information
is absent, optimizing its reference marginal is vacuous, and the i.i.d.
specialization of our result recovers Nomura's second-order formula.

Throughout, whenever we condition on a finite-valued random variable, its
alphabet is identified with its support.  Thus, in particular,
$P_Y(y)>0$ for every $y\in\cY$; zero-mass values are removed without changing
any operational quantity.  For a memoryless source with side information,
define the conditional surprisal and conditional entropy by
\begin{equation}
 \imath_{X|Y}(x|y):=-\log P_{X|Y}(x|y),
 \qquad
 H:=\E_{XY}[\imath_{X|Y}(X|Y)]=H(X|Y)_P.
 \label{eq:conditional-surprisal-entropy}
\end{equation}
All logarithms are base two.  The fundamental first-order statement is that
the amount of asymptotically perfect randomness extractable from $n$
independent copies is $nH+o(n)$.  At this level the answer is governed by the
conditional entropy and is independent of the particular standard security
criterion used to measure leakage and express convergence to the ideal law.

Refinements that quantify the speed of this convergence can depend on the
security criterion.  The relevant fluctuations of the conditional surprisal are
\begin{align}
 V&:=\Var_{XY}[\imath_{X|Y}(X|Y)],
 \label{eq:varentropy-total}\\
 V_1&:=\Var_Y\!\left[\E_{X|Y}[\imath_{X|Y}(X|Y)]\right],
 \label{eq:varentropy-between}\\
 V_2&:=\E_Y\!\left[\Var_{X|Y}[\imath_{X|Y}(X|Y)]\right].
 \label{eq:varentropy-within}
\end{align}
The law of total variance gives $V=V_1+V_2$.  For total variation distance, the
exact second-order expansion is well known: if
$\ell_{\rm TV}^\uparrow(\delta;P_{XY}^{\times n})$ denotes the largest extractable output
length with leakage at most $\delta\in(0,1)$, then
\begin{equation}
 \ell_{\rm TV}^\uparrow(\delta;P_{XY}^{\times n})
 =nH+\sqrt{nV}\,\Phi^{-1}(\delta)+o(\sqrt n),
 \label{eq:tv-second-order-intro}
\end{equation}
where $\Phi^{-1}$ is the standard Gaussian quantile.  Without side
information, Hayashi obtained the corresponding second-order
intrinsic-randomness expansion \cite{hayashi_2008}.  With side
information and total variation, Watanabe and Hayashi obtained the sharper
$O(\log n)$ remainder in the nondegenerate case
\cite[Theorem~3]{watanabe_hayashi_2013}.  If $V=0$, the conditional
surprisal is constant on its support, and two-universal hashing together with
the support-size converse gives
$\ell_{\rm TV}^\uparrow(\delta;P_{XY}^{\times n})=nH+O(1)$; hence
\eqref{eq:tv-second-order-intro} also holds in this case.

Total variation gives statistical indistinguishability between the real and
ideal experiments.  More generally, every $f$-divergence is nonincreasing
under subsequent classical processing by the data-processing inequality, so
$f$-divergences form a natural broader class of security criteria.  Together
with simple monotone transformations, this framework covers both total
variation distance and R\'enyi divergences of orders $\alpha\in(0,1)$.  The
resulting extraction criteria have the same first-order asymptotics but need
not share their refined asymptotics.

Let $f:[0,\infty)\to\mathbb R$ be continuous and convex, and suppose that\footnote{More generally, if $f$ satisfies
$\lim_{t\to\infty} f(t)/t=c<\infty$, one may replace $f$ by the equivalent generator
$\tilde f(t):=f(t)-c(t-1)$. 
This affine transformation leaves the corresponding $f$-divergence unchanged, i.e.,~$D_{\tilde{f}}(P\|Q)=D_f(P\|Q)$, while $\lim_{t\to\infty}\tilde{f}(t)/t=0$~\cite{nomura_yagi_2024}.}
\begin{equation}
 f(1)=0,
 \qquad
 \lim_{t\to\infty}\frac{f(t)}{t}=0.
 \label{eq:f-assumptions}
\end{equation}
These assumptions force $f$ to be nonincreasing, though not necessarily
strictly decreasing.  We call a generator satisfying all these conditions
\emph{admissible}.  For probability distributions on a common finite
alphabet, define the $f$-divergence \cite{csiszar_1967}
\begin{equation}
 D_f(P\|Q):=\sum_{\omega:Q(\omega)>0}
 Q(\omega)f\!\left(\frac{P(\omega)}{Q(\omega)}\right).
 \label{eq:f-divergence}
\end{equation}
The second condition in \eqref{eq:f-assumptions} makes the contribution of
points with $Q(\omega)=0$ vanish, so \eqref{eq:f-divergence} is the complete
formula.  Moreover, $0\le f(q)\le f(0)$ for $q\in[0,1]$.  The generator need
not be pointwise nonnegative for $t>1$, while $D_f(P\|Q)\ge0$ as usual.

For a probability law $P_{ZYS}$, define the fixed- and optimized-marginal
criteria by
\begin{align}
 D_f^\uparrow(P_{ZYS})
 &:=D_f(P_{ZYS}\|U_Z\times P_{YS}),
 \label{eq:security-criteria-up}\\
 D_f^\downarrow(P_{ZYS})
 &:=\inf_{Q_{YS}}D_f(P_{ZYS}\|U_Z\times Q_{YS}).
 \label{eq:security-criteria-down}
\end{align}
where the infimum is over all probability distributions on the adversary's
view $(Y,S)$.  For the joint law induced by the extraction procedure,
$P_{YS}=P_Y\times P_S$, so the fixed-marginal criterion $D_f^\uparrow$
compares the real experiment directly with the ideal law
\eqref{eq:ideal-law}.  The optimized criterion $D_f^\downarrow$ instead
measures divergence to the entire set of laws in which $Z$ is uniform and
independent of the adversary's view.  Proximity to any law in this set
witnesses small leakage, and the closest such law need not have the actual
marginal $P_{YS}$.  Thus $D_f^\downarrow$ is also a natural security
criterion, and it is never larger than $D_f^\uparrow$.  Restricting its
infimum to laws of the form $Q_{YS}=Q_Y\times P_S$ does not change its optimal
second-order limit: achievability uses such a product reference, whereas the
converse allows an arbitrary joint reference law $Q_{YS}$.

\subsection{Main findings and discussion}

Let $P_{XY}$ again denote the joint law of one source symbol $X$ and the
adversary's side information $Y$.  For $n$ independent copies, the source law
is $P_{XY}^{\times n}$.  Write $M_n=|\cZ_n|$ and consider the central-limit
scale
\begin{equation}
 \log M_n=nH+\sqrt n\,L+o(\sqrt n),
 \qquad L\in\mathbb R.
 \label{eq:second-order-rate}
\end{equation}
For an $M_n$-point output, $d_f^\uparrow(M_n;P_{XY}^{\times n})$ and
$d_f^\downarrow(M_n;P_{XY}^{\times n})$ denote the smallest achievable values
of $D_f^\uparrow$ and $D_f^\downarrow$, respectively, where in both cases the
optimization is over seeded hash functions.  We are interested in the exact
$n\to\infty$ limits of these quantities under
\eqref{eq:second-order-rate}; in general, they need not vanish.

The difference between the $D_f^\uparrow$ and $D_f^\downarrow$ criteria is
already visible in the form of the answer.  To present the generic case
without endpoint qualifications, suppose for now that $V_1,V_2>0$; the
formal result below also covers the endpoints.  Write $\Phi$ for the standard
normal distribution function, let $G\sim\mathcal N(0,V_1)$, and, for
$g\in\mathbb R$, define
\begin{equation}
 Q_L(g):=\Phi\!\left(\frac{g-L}{\sqrt{V_2}}\right),
 \qquad\text{so that}\qquad
 \E[Q_L(G)]=\Phi\!\left(-\frac{L}{\sqrt V}\right).
 \label{eq:informal-gaussian-profile}
\end{equation}
Our two main limits are
\begin{align}
 d_f^\uparrow(M_n;P_{XY}^{\times n})
 &\longrightarrow \E[f(Q_L(G))],
 \label{eq:informal-main-limit-up}\\
 d_f^\downarrow(M_n;P_{XY}^{\times n})
 &\longrightarrow f(\E[Q_L(G)])
 =f\!\left(\Phi\!\left(-\frac{L}{\sqrt V}\right)\right).
 \label{eq:informal-main-limit-down}
\end{align}
Two independent Gaussian limits describe the two levels of fluctuation.  After
centering at $nH$ and dividing by $\sqrt n$, the conditional mean of the
block surprisal $\sum_{i=1}^n\imath_{X|Y}(X_i|Y_i)$ converges to $G$.  Within
a fixed fibre $Y^n=y^n$, the remaining fluctuation of the surprisal about
that conditional mean converges, after division by $\sqrt n$, to an
independent $G'\sim\mathcal N(0,V_2)$.  Thus
$Q_L(g)=\Prb\{g+G'\ge L\}$ is the conditional tail probability when the
across-fibre fluctuation has value $g$.  Averaging over $G$ removes this
conditioning and performs a Gaussian convolution.  Since $G+G'$ is normal
with variance $V_1+V_2=V$, this gives the expectation identity in
\eqref{eq:informal-gaussian-profile}.

The two limits differ only in the order of the nonlinear map $f$ and this
average over fibres.  For $D_f^\uparrow$, $f$ is applied to each conditional
tail probability before the fibres are averaged, producing
$\E[f(Q_L(G))]$.  Optimizing the reference marginal in $D_f^\downarrow$
reweights the fibres so that their tail probabilities are aggregated before
$f$ is applied, producing $f(\E[Q_L(G)])$.  Thus the $D_f^\downarrow$ limit
depends only on the total varentropy $V$, whereas the $D_f^\uparrow$ limit
generally remembers its split into $V_1$ and $V_2$.  Jensen's inequality gives
\begin{equation}
 f\!\left(\Phi\!\left(-\frac{L}{\sqrt V}\right)\right)
 \leq \E[f(Q_L(G))].
 \label{eq:informal-jensen-comparison}
\end{equation}
Appendix~\ref{app:marginal-comparison} gives a technical explanation of this
distinction.  If $f$ is affine on the relevant interval, Jensen's inequality
in \eqref{eq:informal-jensen-comparison} is an equality and the two profiles
agree.  Total variation is the principal example, so optimizing its reference
marginal does not alter the second-order limit.  The resulting corollary
recovers the known side-information formula cited above, and we
make no novelty claim for that specialization.


For the power generators corresponding to R\'enyi divergences of order
$\alpha\in(0,1)$, a monotone logarithmic transformation turns
$D_f^\uparrow$ and $D_f^\downarrow$ into the two R\'enyi criteria studied by
Hayashi and Tan \cite{hayashi_tan_2017}; see also
\cite{hayashi_tan_correction_2024}.  They call the latter the
``Gallager form'' because the associated optimized
conditional R\'enyi entropy is a rescaled classical Gallager function
\cite[Eqs.~(14)--(17)]{hayashi_tan_2017}; see also \cite{gallager_1968}.
In the second-order regime,
their Eqs.~(73) and (74) give nonmatching bounds for the two criteria rather
than exact fixed-$L$ limits.  Our R\'enyi corollary evaluates both limits
exactly and closes these gaps.

\subsection{Related work}

For quantum side information, Tomamichel and Hayashi first obtained a tight
second-order expansion under purified distance after optimizing the reference
side-information state \cite{tomamichel_hayashi_2013}.  Up to the usual
monotone transformation between purified distance, fidelity, and
$D_{1/2}$, this is the optimized order-$1/2$ R\'enyi-divergence criterion.
Their result did not settle the corresponding fixed-marginal problem.  Our
results help explain this difficulty: even for commuting side
information, the fixed-marginal order-$1/2$ criterion has a Gaussian-mixture
profile rather than the single Gaussian profile of the optimized criterion.
For trace distance, Shen, Gao, and Cheng proved a strong converse and
moderate-deviation bounds \cite{shen_gao_cheng_2022} and subsequently obtained
the tight second-order expansion for the standard fixed-marginal
privacy-amplification criterion \cite{shen_gao_cheng_2024}.  At the one-shot
level, Anshu, Berta, Jain, and Tomamichel derived privacy-amplification bounds
in terms of the partially smoothed conditional min-entropy
\cite{anshu_berta_jain_tomamichel_2020}, while Regula and Tomamichel recently
obtained tighter, approximately matching trace-distance bounds
\cite{regula_tomamichel_2026}.  Renes gave hypothesis-testing
characterizations for the trace-distance privacy-amplification criterion in
the present setting \cite{renes_2018}; Abdelhadi and Renes later showed that the
second-order term of this partially smoothed conditional min-entropy is not
uniform across quantum states and, for pure states, differs from the usual
globally smoothed Gaussian form \cite{abdelhadi_renes_2020}.

Recent work also clarifies the information quantities governing other
asymptotic regimes.  Li, Li, and Yu systematically studied a two-parameter
conditional R\'enyi entropy that, after a change of variables, coincides with
the conditional quantity of Hayashi and Tan and arises as the classical
specialization of the three-parameter quantum conditional entropies introduced
by Rubboli, Goodarzi, and Tomamichel
\cite{li_li_yu_2025,hayashi_tan_2017,rubboli_goodarzi_tomamichel_2024}.  They introduced the
associated two-parameter mutual information and used these quantities to
characterize strong-converse exponents for privacy amplification and soft
covering under R\'enyi criteria \cite{li_li_yu_2025}.  Berta and Yao
independently determined the strong-converse exponent for classical privacy
amplification under purified distance; their result is closely related to the
order-$1/2$ instance of the R\'enyi-divergence exponent
\cite{berta_yao_2025}.  Li and Yao gave an
operational interpretation of sandwiched R\'enyi divergences of orders
between $1/2$ and one as strong-converse exponents under fidelity or purified
distance \cite{li_yao_2024}.  These strong-converse results, expressed through
a variety of R\'enyi-type entropies, are compatible with the refined
second-order limits obtained here: they describe rates fixed beyond the
first-order boundary, whereas our profiles resolve the
$nH+L\sqrt n+o(\sqrt n)$ transition window around that boundary.
Appendix~\ref{app:strong-converse-compatibility} makes this compatibility
explicit for the fixed-marginal exponent of Li, Li, and Yu.
At a complementary axiomatic level, Rubboli,
Haapasalo, and Tomamichel completely characterized classical conditional
entropies satisfying additivity, relabeling invariance, and monotonicity
under conditional mixing \cite{rubboli_haapasalo_tomamichel_2026}.  Their
characterization identifies exponential averages of conditional R\'enyi
entropies as the general family, providing structural context for the
conditional R\'enyi quantities considered here.

\section{Formal Statement of Main Results}
\label{sec:main-result}

This section states the exact second-order limits, explains the resulting
Gaussian profiles, and derives the R\'enyi and total-variation corollaries.

\subsection{Optimal leakage and the Gaussian \texorpdfstring{$f$}{f}-profile}

For $M\in\mathbb N$ and $\star\in\{\downarrow,\uparrow\}$, define the
optimal leakage by
\begin{equation}
 d_f^\star(M;P_{XY}):=\inf_{(P_S,\{\varphi_s\})}
 D_f^\star(P_{\varphi_S(X)YS})
 \label{eq:optimal-leakage}
\end{equation}
where $S$ is finite, $S\perp(X,Y)$, and every $\varphi_s$ maps $\cX$ to the
$M$-point alphabet $\cZ$.  The optimization in $d_f^\downarrow$ includes the
reference law on $(Y,S)$.  Since the actual
marginal $P_Y\times P_S$ is always feasible, at every blocklength
\begin{equation}
 d_f^\downarrow(M;P_{XY})\le d_f^\uparrow(M;P_{XY}).
 \label{eq:criterion-order}
\end{equation}

For $0\le r<1$ and $g,x\in\mathbb R$, set
\begin{equation}
 Q_r(g;x):=\Phi\!\left(\frac{x+\sqrt r\,g}{\sqrt{1-r}}\right),
 \label{eq:Qr}
\end{equation}
and introduce the Gaussian $f$-profile
\begin{equation}
 \Fprof_{f,r}(x):=\E[f(Q_r(G;x))],
 \qquad G\sim N(0,1).
 \label{eq:F-profile}
\end{equation}
At $r=0$, the argument of $f$ is the constant $\Phi(x)$.  To obtain the
other endpoint, fix $g\ne-x$.  As $r\uparrow1$, the denominator in
\eqref{eq:Qr} tends to zero, while the numerator tends to $x+g$.  Hence
$Q_r(g;x)\to\ind\{g>-x\}$.  Since $f$ is bounded on $[0,1]$, dominated
convergence gives
$f(0)\Prb\{G<-x\}+f(1)\Prb\{G>-x\}=f(0)\Phi(-x)$, where we used
$f(1)=0$.  Thus, defining the value at $r=1$
by continuous extension, the two endpoint formulas are
\begin{equation}
 \Fprof_{f,0}(x)=f(\Phi(x)),
 \qquad
 \Fprof_{f,1}(x)=f(0)\Phi(-x).
 \label{eq:F-endpoints}
\end{equation}
Lemma~\ref{lem:profile-regularity} below shows that the profile is continuous
and nonincreasing from $f(0)$ to zero.  When $f$ is strictly decreasing on
$(0,1)$, so is the profile.

\begin{theorem}[Exact second-order $f$-leakage: fixed and optimized marginals]
\label{thm:main-f}
Let $P_{XY}$ have finite support, let $V>0$, and fix an admissible generator
$f$ and $L\in\mathbb R$.  For every integer sequence
$M_n$ satisfying \eqref{eq:second-order-rate},
\begin{align}
 \lim_{n\to\infty}d_f^\uparrow(M_n;P_{XY}^{\times n})
 &=\Fprof_{f,V_1/V}\!\left(-\frac{L}{\sqrt V}\right),
 \label{eq:main-f-theorem}\\
 \lim_{n\to\infty}d_f^\downarrow(M_n;P_{XY}^{\times n})
 &=f\!\left(\Phi\!\left(-\frac{L}{\sqrt V}\right)\right)
 =\Fprof_{f,0}\!\left(-\frac{L}{\sqrt V}\right).
 \label{eq:main-f-down-theorem}
\end{align}
Two-universal hashing achieves both limits, and both converses hold for
arbitrary seeded hash families.  For the optimized criterion, the converse
even permits the reference law to be optimized jointly on $(Y^n,S_n)$.
\end{theorem}

The profile always lies between two bounds determined by the total Gaussian
tail.  The identity $\E Q_r(G;x)=\Phi(x)$ holds for every $r\in[0,1]$, with
$r=1$ interpreted through the limiting indicator above.  Jensen's inequality
gives the lower bound below, while convexity and $f(1)=0$ give
$f(q)\le(1-q)f(0)$ for $q\in[0,1]$, and hence the upper bound:
\begin{equation}
 f(\Phi(x))\le \Fprof_{f,r}(x)
 \le f(0)\Phi(-x).
 \label{eq:f-profile-envelope}
\end{equation}
The optimized-marginal limit in \eqref{eq:main-f-down-theorem} equals the
lower bound, consistently with \eqref{eq:criterion-order}.  For the
fixed-marginal limit, the lower bound is attained when $V_1=0$, while the
upper bound is attained when $V_2=0$.  If $f$ is affine on $[0,1]$, the two
bounds coincide and only the total variance remains.

\subsection{Fixed-leakage expansions}

We first note a monotonicity property that will also be used in the proof.
For either arrow, the sequence
$\ell\mapsto d_f^\star(2^\ell;P_{XY})$ is nondecreasing.  Indeed, identify a
$2M$-point output alphabet with an $M$-point alphabet times one bit and
discard the latter.  The balanced projection sends $U_{2M}$ to $U_M$, so
data processing proves the claim for the fixed reference.  For the optimized
reference, apply the same projection to both the distribution induced by the
hash and every
candidate $U_{2M}\times Q_{YS}$; its image is $U_M\times Q_{YS}$, after which
one takes the infimum over $Q_{YS}$ and over hash families.

For a leakage budget $0<\delta<f(0)$ and
$\star\in\{\downarrow,\uparrow\}$, define
\begin{equation}
 \ell_f^\star(\delta;P_{XY})
 :=\sup\{\ell\in\mathbb N_0:d_f^\star(2^\ell;P_{XY})\le\delta\}.
 \label{eq:fixed-f-leakage-def}
\end{equation}
Here the supremum is taken in $\mathbb N_0\cup\{+\infty\}$.  The admissible
set is nonempty because a one-point output has zero leakage.  The proof below
shows that, for the i.i.d. laws in the corollary, it is bounded for every
sufficiently large $n$; its supremum is then attained and is the largest
extractable integer output length.

\begin{corollary}[Fixed-leakage $f$-expansions: fixed and optimized marginals]
\label{cor:fixed-f-leakage}
Under the assumptions of Theorem~\ref{thm:main-f}, suppose in addition that
$f$ is strictly decreasing on $(0,1)$.  For every $\delta\in(0,f(0))$,
\begin{align}
 \ell_f^\uparrow(\delta;P_{XY}^{\times n})
 &=nH-\sqrt{nV}\,
 \Fprof_{f,V_1/V}^{-1}(\delta)+o(\sqrt n).
 \label{eq:f-quantile-expansion}\\
 \ell_f^\downarrow(\delta;P_{XY}^{\times n})
 &=nH-\sqrt{nV}\,\Phi^{-1}\!\left(f^{-1}(\delta)\right)
 +o(\sqrt n).
 \label{eq:f-quantile-down-expansion}
\end{align}
\end{corollary}

\subsection{R\'enyi divergences as a corollary}

For $\alpha\in(0,1)$, let
$f_\alpha(t):=(1-t^\alpha)/(1-\alpha)$.  This is an admissible generator,
with $f_\alpha(0)=(1-\alpha)^{-1}$.  Define the R\'enyi divergence by
\begin{equation}
 D_\alpha(P\|Q):=T_\alpha\!\left(D_{f_\alpha}(P\|Q)\right),
 \label{eq:renyi-transform}
\end{equation}
where, with the value at the right endpoint understood in the extended sense,
\begin{equation}
 T_\alpha(u):=-\frac{1}{1-\alpha}\log\bigl(1-(1-\alpha)u\bigr).
 \label{eq:renyi-transform-function}
\end{equation}
The transformation in \eqref{eq:renyi-transform-function} is continuous and
strictly increasing on $[0,(1-\alpha)^{-1})$, with the extended value
$+\infty$ at the right endpoint.  Introduce the Gaussian
R\'enyi overlap profile
\begin{equation}
 \Gprof_{\alpha,r}(x):=\E[Q_r(G;x)^\alpha]\ (0<r<1),\qquad
 \Gprof_{\alpha,0}(x):=\Phi(x)^\alpha,\quad
 \Gprof_{\alpha,1}(x):=\Phi(x).
 \label{eq:G-profile}
\end{equation}
For every $r\in[0,1]$, this profile is continuous and strictly increasing from
zero to one.

For $\star\in\{\downarrow,\uparrow\}$, let $D_\alpha^\star$ denote the
criterion obtained from $D_f^\star$ by replacing $D_f$ with $D_\alpha$, and
let $d_\alpha^\star$ denote the corresponding optimal leakage from
\eqref{eq:optimal-leakage}.
Because $T_\alpha$ is continuous and increasing, it commutes with the two
infima, first over the reference law and then over seeded hash families.
Consequently, at every blocklength,
\begin{equation}
 d_\alpha^\star(M;P_{XY})
 =T_\alpha\!\left(d_{f_\alpha}^\star(M;P_{XY})\right).
 \label{eq:renyi-operational-transform}
\end{equation}

\begin{corollary}[R\'enyi specialization: fixed and optimized marginals]
\label{cor:renyi}
Let $P_{XY}$ have finite support, let $V>0$, and fix
$\alpha\in(0,1)$ and $L\in\mathbb R$.  Let
$d_\alpha^\uparrow(M;P_{XY})$ and $d_\alpha^\downarrow(M;P_{XY})$ be the two
optimal leakages defined above.  If $M_n$ satisfies
\eqref{eq:second-order-rate}, then
\begin{align}
 \lim_{n\to\infty}d_\alpha^\uparrow(M_n;P_{XY}^{\times n})
 &=-\frac{1}{1-\alpha}\log
 \Gprof_{\alpha,V_1/V}\!\left(-\frac{L}{\sqrt V}\right),
 \label{eq:renyi-main}\\
 \lim_{n\to\infty}d_\alpha^\downarrow(M_n;P_{XY}^{\times n})
 &=-\frac{\alpha}{1-\alpha}\log
 \Phi\!\left(-\frac{L}{\sqrt V}\right).
 \label{eq:renyi-down-main}
\end{align}
For $\delta>0$, let $\ell_\alpha^\uparrow(\delta;P_{XY})$ and
$\ell_\alpha^\downarrow(\delta;P_{XY})$ denote the largest extractable output
lengths under the fixed- and optimized-marginal criteria, respectively.  Then
\begin{align}
 \ell_\alpha^\uparrow(\delta;P_{XY}^{\times n})
 &=nH-\sqrt{nV}\,
 \Gprof^{-1}_{\alpha,V_1/V}
 \!\left(2^{-(1-\alpha)\delta}\right)+o(\sqrt n),
 \label{eq:renyi-fixed-leakage}\\
 \ell_\alpha^\downarrow(\delta;P_{XY}^{\times n})
 &=nH-\sqrt{nV}\,
 \Phi^{-1}\!\left(2^{-(1-\alpha)\delta/\alpha}\right)
 +o(\sqrt n).
 \label{eq:renyi-down-fixed-leakage}
\end{align}
In both cases two-universal hashing achieves the stated limits, and the
converses permit arbitrary seeded hash families.
\end{corollary}

\begin{proof}
We first make the reference-marginal optimization explicit.  Define
\begin{equation}
 A_\alpha(y,s):=M^{\alpha-1}\sum_zP_{Z|YS}(z|y,s)^\alpha.
 \label{eq:renyi-conditional-overlap}
\end{equation}
Minimizing $D_\alpha(P_{ZYS}\|U_Z\times Q_{YS})$ is equivalent to maximizing
$\sum_{y,s}P_{YS}(y,s)^\alpha Q_{YS}(y,s)^{1-\alpha}A_\alpha(y,s)$.
H\"older's inequality, or the Lagrange multiplier equations, gives
\begin{align}
 Q_{YS}^*(y,s)
 &=\frac{P_{YS}(y,s)A_\alpha(y,s)^{1/\alpha}}
 {\sum_{y',s'}P_{YS}(y',s')A_\alpha(y',s')^{1/\alpha}},
 \label{eq:renyi-tilted-reference}\\
 D_\alpha^\downarrow(P_{ZYS})
 &=-\frac{\alpha}{1-\alpha}\log
 \sum_{y,s}P_{YS}(y,s)A_\alpha(y,s)^{1/\alpha}.
 \label{eq:renyi-down-exact}
\end{align}
The fixed-reference criterion follows immediately from
\eqref{eq:renyi-transform}, and optimization over seeded hash families gives
\eqref{eq:renyi-operational-transform}.

The two leakage limits follow from the corresponding statements of
Theorem~\ref{thm:main-f} through
\eqref{eq:renyi-operational-transform}: substituting the power generator and
applying $T_\alpha$ gives the profile \eqref{eq:G-profile}.  For a R\'enyi leakage budget
$\delta>0$, the equivalent power-divergence budget is
\begin{equation}
 \delta_f=T_\alpha^{-1}(\delta)
 =\frac{1-2^{-(1-\alpha)\delta}}{1-\alpha}.
 \label{eq:renyi-budget-transform}
\end{equation}
Applying Corollary~\ref{cor:fixed-f-leakage} with this budget and using
$1-(1-\alpha)f_\alpha(q)=q^\alpha$ gives both fixed-leakage expansions.
\end{proof}

For comparison, the fixed-marginal part of Hayashi and Tan's result
\cite[Eq.~(73)]{hayashi_tan_2017} gives the envelope
\begin{align}
 -\frac{\alpha}{1-\alpha}\log\Phi\!\left(-\frac{L}{\sqrt V}\right)
 &\le \liminf_{n\to\infty}d_\alpha^\uparrow(M_n;P_{XY}^{\times n}),
 \label{eq:HT-lower}\\
 \limsup_{n\to\infty}d_\alpha^\uparrow(M_n;P_{XY}^{\times n})
 &\le-\frac{1}{1-\alpha}\log
 \Phi\!\left(-\frac{L}{\sqrt V}\right).
 \label{eq:HT-upper}
\end{align}
Since $q\le q^\alpha$ and $q\mapsto q^\alpha$ is concave on $[0,1]$, one has
$\Phi(x)\le\Gprof_{\alpha,r}(x)\le\Phi(x)^\alpha$, so
\eqref{eq:renyi-main} lies exactly inside
\eqref{eq:HT-lower}--\eqref{eq:HT-upper}.

Their Gallager-form criterion coincides with $D_\alpha^\downarrow$, up to
notation.  In Case~(C), their Eq.~(74) again provides nonmatching
second-order bounds, whereas \eqref{eq:renyi-down-main} supplies the exact
fixed-$L$ limit \cite{hayashi_tan_2017}; see also
\cite{hayashi_tan_correction_2024}.

Fig.~\ref{fig:profiles} displays the R\'enyi profile and its inverse at
$\alpha=1/2$.  It also makes visible that a single total variance does not
determine the second-order penalty for a nonlinear generator.
\begin{figure}[H]
 \centering
 \includegraphics[width=0.94\textwidth]{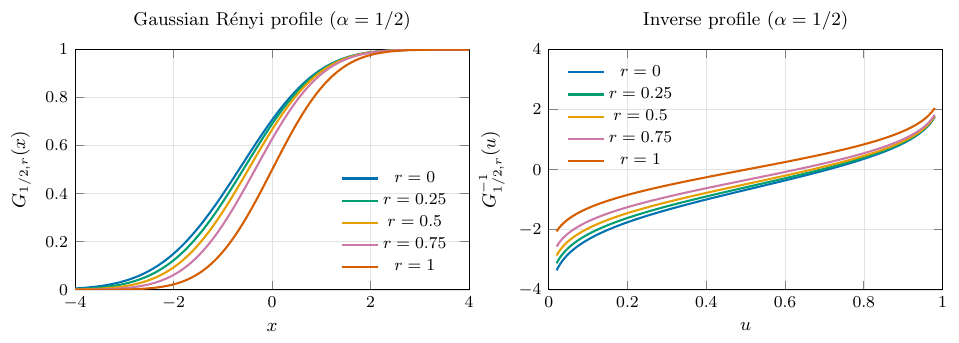}
 \caption{The Gaussian R\'enyi overlap profile and its inverse for
 $\alpha=1/2$.  The variance ratio is $r=V_1/(V_1+V_2)$.  In the
 fixed-leakage expansion, the inverse is evaluated at
 $u=2^{-\delta/2}$.  The endpoints are
 $\Gprof_{1/2,0}(x)=\sqrt{\Phi(x)}$ and
 $\Gprof_{1/2,1}(x)=\Phi(x)$.}
 \label{fig:profiles}
\end{figure}

\paragraph{Total variation as a recovered corollary.}
The usual generator of total variation is $|t-1|/2$.  For our convention it
is convenient to subtract the affine term $(t-1)/2$ and use the equivalent
generator
\begin{equation}
 f_{\rm TV}(t):=\frac{|t-1|}{2}-\frac{t-1}{2}=(1-t)_+.
 \label{eq:tv-generator}
\end{equation}
Indeed, directly from \eqref{eq:f-divergence},
\begin{equation}
 D_{f_{\rm TV}}(P\|Q)
 =\frac12\sum_\omega|P(\omega)-Q(\omega)|
 =:\|P-Q\|_{\rm TV}.
 \label{eq:tv-divergence}
\end{equation}
Thus $f_{\rm TV}$ is admissible, with $f_{\rm TV}(0)=1$.

Let $d_{\rm TV}^\uparrow(M;P_{XY})$ and
$d_{\rm TV}^\downarrow(M;P_{XY})$ denote \eqref{eq:optimal-leakage} with
total variation in place of $D_f$.  Since $f_{\rm TV}$ is affine on
$[0,1]$, Gaussian convolution
gives, for every $r\in[0,1]$,
\begin{align}
 \Fprof_{f_{\rm TV},r}(x)
 &=\E[1-Q_r(G;x)]=\Phi(-x),
 \label{eq:tv-fixed-profile}\\
 f_{\rm TV}(\Phi(x))&=\Phi(-x).
 \label{eq:tv-down-profile}
\end{align}
The fixed and optimized reference marginals therefore have the same exact
second-order leakage, even though their finite-blocklength values need not
coincide.  We recover the result in
\cite{watanabe_hayashi_2013,hayashi_2016,anshu_berta_jain_tomamichel_2020}.

\begin{corollary}[Total variation]
\label{cor:total-variation}
Let $P_{XY}$ have finite support and $V>0$.  If $M_n$ satisfies
\eqref{eq:second-order-rate}, then
\begin{equation}
 \lim_{n\to\infty}d_{\rm TV}^\uparrow(M_n;P_{XY}^{\times n})
 =\lim_{n\to\infty}d_{\rm TV}^\downarrow(M_n;P_{XY}^{\times n})
 =\Phi\!\left(\frac{L}{\sqrt V}\right).
 \label{eq:tv-main}
\end{equation}
For every $\delta\in(0,1)$, the fixed-marginal and optimized-marginal output
lengths both obey
\begin{align}
 \ell_{\rm TV}^\uparrow(\delta;P_{XY}^{\times n})
 &=nH+\sqrt{nV}\,\Phi^{-1}(\delta)+o(\sqrt n),
 \label{eq:tv-fixed-length}\\
 \ell_{\rm TV}^\downarrow(\delta;P_{XY}^{\times n})
 &=nH+\sqrt{nV}\,\Phi^{-1}(\delta)+o(\sqrt n).
 \label{eq:tv-down-length}
\end{align}
Two-universal hashing achieves these limits, and the converses permit
arbitrary seeded hash families.
\end{corollary}

Affine generators therefore guarantee this coincidence.  For nonlinear
$f$, Jensen's inequality is generally strict: fixing the
reference retains the conditional variance split, whereas optimizing it
depends only on the unconditional surprisal tail.

\section{One-Shot Bounds}
\label{sec:oneshot}

This section establishes the finite-blocklength converse and achievability
bounds used in the asymptotic proof.

For a probability vector $p=(p_x)$ on its support and $c>0$, define
\begin{align}
 \mathcal Q_c&:=\left\{q:0\le q_x\le c,\ \sum_xq_x\le1\right\},
 \label{eq:Qc}\\
 d_f(p,q)&:=\sum_xq_xf\!\left(\frac{p_x}{q_x}\right)
 +\left(1-\sum_xq_x\right)f(0),
 \label{eq:df-perspective}\\
 g_{f,c}(p)&:=\min_{q\in\mathcal Q_c}d_f(p,q).
 \label{eq:gdef}
\end{align}
The vector $q$ is a subprobability approximation to $p$ whose individual
coordinates are capped by $c$.  In the application below, $c=2^{-h}$, and
$h=\log M$ gives the cap $1/M$ associated with a uniform law on $M$ output
symbols.  More concretely, pulling that uniform law back through a hash
assigns total reference mass $1/M$ to each nonempty output bin and therefore
at most $1/M$ to any individual input point.

The quantity $d_f(p,q)$ is the resulting perspective cost.  Since $q$ may
have total mass smaller than one, its final term completes the reference law
by placing the unused mass on zero-probability dummy symbols.  Thus
$g_{f,c}(p)$ is the smallest $f$-divergence cost compatible with the cap $c$.
The perspective term at $q_x=0$ is defined by continuity as
$p_x\lim_{t\to\infty}f(t)/t=0$.  For a joint distribution $P_{XY}$ and
threshold $h$, set
\begin{equation}
 \Gamma_{f,P}^\uparrow(h):=\E_Y[g_{f,2^{-h}}(P_{X|Y})].
 \label{eq:Gamma}
\end{equation}
This averages the optimal capped cost separately in every conditional fibre,
using the actual marginal $P_Y$, and is therefore the one-shot quantity for
the fixed-reference criterion.

For $a\ge0$, we also need the scaled capped value
\begin{equation}
 g_{f,c}^{(a)}(p):=\min_{q\in\mathcal Q_c}
 \left\{\sum_xq_xf\!\left(\frac{ap_x}{q_x}\right)
 +\left(1-\sum_xq_x\right)f(0)\right\}.
 \label{eq:scaled-gdef}
\end{equation}
Thus $g_{f,c}^{(1)}=g_{f,c}$.  For a joint law $P_{XY}$, define
\begin{equation}
 \Gamma_{f,P}^\downarrow(h)
 :=\inf_{Q_Y}\sum_{y:Q_Y(y)>0}Q_Y(y)
 g_{f,2^{-h}}^{(P_Y(y)/Q_Y(y))}(P_{X|Y=y}),
 \label{eq:Gamma-down}
\end{equation}
where terms with $Q_Y(y)=0$ use the same reference-zero convention as
\eqref{eq:f-divergence}.
Here the scale $a=P_Y(y)/Q_Y(y)$ accounts for the change in likelihood ratio
when the reference marginal of a fibre is changed from $P_Y(y)$ to $Q_Y(y)$.
The quantity $\Gamma_{f,P}^\downarrow$ then chooses the reference marginal
$Q_Y$ that gives the smallest weighted sum of these scaled fibre costs.  The
arrows on the two $\Gamma$ quantities thus match the corresponding security
criteria.

The next lemma converts the output-bin masses into a one-shot converse that
applies to every seeded family.

\begin{lemma}
\label{lem:oneshot-converse}
\label{lem:oneshot-converse-down}
Every seeded family with output size $M$ satisfies the fixed-reference bound
$D_f^\uparrow(P_{\varphi_S(X)YS})\ge\Gamma_{f,P}^\uparrow(\log M)$ and the
optimized-reference bound
$D_f^\downarrow(P_{\varphi_S(X)YS})\ge\Gamma_{f,P}^\downarrow(\log M)$.
Consequently,
\begin{equation}
 d_f^\uparrow(M;P_{XY})\ge\Gamma_{f,P}^\uparrow(\log M)
 \qquad\text{and}\qquad
 d_f^\downarrow(M;P_{XY})\ge\Gamma_{f,P}^\downarrow(\log M).
 \label{eq:optimal-one-shot-converses}
\end{equation}
\end{lemma}

\begin{proof}
Fix $s$ and $y$ with $P_Y(y)>0$, abbreviate
$p_x=P_{X|Y}(x|y)$, and let
\begin{equation}
 r_z:=\sum_{x:\varphi_s(x)=z}p_x
\end{equation}
be the output-bin masses.  For $p_x>0$, define
\begin{equation}
 q_{s,y}(x):=\frac{p_x}{Mr_{\varphi_s(x)}}.
 \label{eq:q-from-bins}
\end{equation}
The denominator is positive because the bin contains $x$.  Since
$p_x\le r_{\varphi_s(x)}$, one has $q_{s,y}(x)\le1/M$.  Summing bin by bin,
\begin{equation}
 \sum_xq_{s,y}(x)=\frac{|\{z:r_z>0\}|}{M}\le1,
\end{equation}
so $q_{s,y}\in\mathcal Q_{1/M}$.  Moreover,
\begin{align}
 &\sum_xq_{s,y}(x)f\!\left(\frac{p_x}{q_{s,y}(x)}\right)
 +\left(1-\sum_xq_{s,y}(x)\right)f(0)\\
 &\qquad=\frac1M\sum_z f(Mr_z)
 =D_f(P_{\varphi_s(X)|Y=y}\|U_Z).
 \label{eq:bin-identity}
\end{align}
The term involving $f(0)$ accounts exactly for the empty bins.  The left side
of \eqref{eq:bin-identity} is at least $g_{f,1/M}(p)$.  Averaging over
$(S,Y)$ proves the fixed-reference family inequality; taking the infimum
proves the corresponding bound for $d_f^\uparrow$.

For the optimized claim, fix a reference $Q_{YS}$ and let
$a_{ys}=P_Y(y)P_S(s)/Q_{YS}(y,s)$ whenever the denominator is positive.
The bin construction in \eqref{eq:q-from-bins}, with the conditional source
scaled by $a_{ys}$, gives
\begin{align}
 &D_f(P_{\varphi_S(X)YS}\|U_Z\times Q_{YS})\\
 &\quad\ge\sum_{y,s:Q_{YS}(y,s)>0}Q_{YS}(y,s)
 g_{f,1/M}^{(a_{ys})}(P_{X|Y=y}).
 \label{eq:scaled-bin-lower-bound}
\end{align}
For fixed $y$, aggregate the feasible capped vectors over $s$ with weights
proportional to $Q_{YS}(y,s)$.  Their average is still capped by $1/M$.
Put
\begin{equation}
 \theta_y:=\sum_{s:Q_{YS}(y,s)>0}P_S(s)\le1.
 \label{eq:theta-reference-support}
\end{equation}
Joint convexity of the perspective $(u,v)\mapsto vf(u/v)$ first shows that
the corresponding sum in \eqref{eq:scaled-bin-lower-bound} is at least
\begin{equation}
 Q_Y(y)g_{f,1/M}^{(\theta_yP_Y(y)/Q_Y(y))}(P_{X|Y=y}).
 \end{equation}
For every fixed feasible capped vector, its scaled objective is
nonincreasing in the scale because $f$ is nonincreasing.  Taking the minimum
preserves this property, and hence the last display is at least
\begin{equation}
 Q_Y(y)g_{f,1/M}^{(P_Y(y)/Q_Y(y))}(P_{X|Y=y}).
\end{equation}
Here $Q_Y$ is the $Y$-marginal of $Q_{YS}$.  Summing over $y$, then
minimizing first over $Q_{YS}$ and finally over the seeded family, proves the
two claims.
\end{proof}

We use the standard notion of two-universal hashing introduced by Carter and
Wegman \cite{carter_wegman_1979}.

\begin{definition}
\label{def:two-universal}
A seeded family $\{\varphi_s:\cX\to\cZ\}_{s\in\cS}$ is two-universal under
$P_S$ if
\begin{equation}
 \Prb\{\varphi_S(x)=\varphi_S(x')\}\le\frac1M
 \quad\text{for all }x\ne x'.
 \label{eq:two-universal-definition}
\end{equation}
\end{definition}

Such a family exists for every finite input alphabet and every $M$: one may
take the uniform distribution over all maps from the input alphabet to
$\{1,\ldots,M\}$, for which the collision probability is exactly $1/M$.

The following second-moment estimate is the only consequence of
two-universality needed in the direct part.
It says that, when nonnegative input weights are assigned to hash bins, the
expected squared deviation of the bin weights from an equal allocation is
controlled by the sum of the squared input weights.

\begin{lemma}
\label{lem:two-universal-second-moment}
If the family is two-universal, then for any nonnegative numbers $\ell_x$,
\begin{equation}
 \E_S\sum_z\left(
 \sum_{x:\varphi_S(x)=z}\ell_x-\frac1M\sum_x\ell_x
 \right)^2
 \le\left(1-\frac1M\right)\sum_x\ell_x^2.
 \label{eq:universal-second-moment}
\end{equation}
\end{lemma}

\begin{proof}
For the calculation, put $Q:=\sum_x\ell_x$ and
$L_S(z):=\sum_{x:\varphi_S(x)=z}\ell_x$.  For every seed,
$\sum_zL_S(z)=Q$, so expanding the square shows that the left side of
\eqref{eq:universal-second-moment} is
$\E_S\sum_zL_S(z)^2-Q^2/M$.

The square $L_S(z)^2$ contains one diagonal term $\ell_x^2$ for each input
$x$ assigned to $z$, as well as one term $\ell_x\ell_{x'}$ for every ordered
pair of distinct inputs that collide in that bin.  After summing over $z$ and
averaging over the seed, the diagonal contribution is therefore
$\sum_x\ell_x^2$, while two-universality bounds the collision contribution
by $M^{-1}\sum_{x\ne x'}\ell_x\ell_{x'}$.  Since the latter sum equals
$Q^2-\sum_x\ell_x^2$, we obtain
$\E_S\sum_zL_S(z)^2\le
(1-M^{-1})\sum_x\ell_x^2+Q^2/M$.  Subtracting $Q^2/M$ proves the claim.
\end{proof}

For $\tau>0$ and each $y$, retain the light conditional probabilities
\begin{align}
 \ell_y(x)&:=P_{X|Y}(x|y)
 \ind\!\left\{P_{X|Y}(x|y)\le\frac{2^{-\tau}}M\right\},
 \label{eq:light-part}\\
 Q_\tau(y)&:=\sum_x\ell_y(x).
 \label{eq:light-mass}
\end{align}
Define the modulus of continuity of $f$ on $[0,1]$ by
\begin{equation}
 \omega_f(u):=\sup\{|f(a)-f(b)|:a,b\in[0,1],\ |a-b|\le u\}.
 \label{eq:f-modulus}
\end{equation}
Here $Q_\tau(y)$ is the total conditional mass of the light symbols, rather
than a new probability distribution.  The cutoff in \eqref{eq:light-part}
places each retained symbol a factor $2^{-\tau}$ below the target mass $1/M$
of a uniform output bin; this slack is what allows two-universal hashing to
spread the retained mass nearly uniformly.  The modulus $\omega_f(u)$ is the
largest possible change in $f$ between two arguments in $[0,1]$ that are at
most $u$ apart.  Its convergence to zero as $u\downarrow0$ turns the resulting
approximation in $\ell_1$ distance into an approximation of the
$f$-divergence cost.

\begin{lemma}
\label{lem:light-achievability}
If $S$ indexes a two-universal family, then
\begin{align}
 &D_f(P_{\varphi_S(X)YS}\|U_Z\times P_Y\times P_S)\\
 &\quad\le \E_Y[f(Q_\tau(Y))]
 +\omega_f(2^{-\tau/4})+f(0)2^{-\tau/4}.
 \label{eq:light-achievability}
\end{align}
Put $\overline Q_\tau:=\E_YQ_\tau(Y)$.  If $\overline Q_\tau>0$, define
\begin{equation}
 Q_{YS}^\circ(y,s)
 :=\frac{P_Y(y)P_S(s)Q_\tau(y)}{\overline Q_\tau}.
 \label{eq:one-shot-tail-tilted-reference}
\end{equation}
Then the same family satisfies
\begin{align}
 &D_f(P_{\varphi_S(X)YS}\|U_Z\times Q_{YS}^\circ)\\
 &\quad\le f(\overline Q_\tau)
 +\omega_f(2^{-\tau/4})+f(0)2^{-\tau/4}.
 \label{eq:light-achievability-down}
\end{align}
\end{lemma}

\begin{proof}
Fix $y$ and let
$L_{S,y}(z)=\sum_{x:\varphi_S(x)=z}\ell_y(x)$.  From
Lemma~\ref{lem:two-universal-second-moment} and the cap in
\eqref{eq:light-part},
\begin{equation}
 \E_S\sum_z\left(L_{S,y}(z)-\frac{Q_\tau(y)}M\right)^2
 \le\frac{2^{-\tau}}M Q_\tau(y).
 \label{eq:L2-light}
\end{equation}
Cauchy--Schwarz, followed by Jensen, yields
\begin{equation}
 \E_S\sum_z\left|L_{S,y}(z)-\frac{Q_\tau(y)}M\right|
 \le2^{-\tau/2}.
 \label{eq:L1-light}
\end{equation}
Let $L_{ZYS}$ be the joint subdistribution produced by hashing the light part,
including $P_S$ and $P_Y$, and let $R_{ZYS}$ have conditional output
$Q_\tau(y)U_Z$.  Averaging \eqref{eq:L1-light} over $Y$ gives
$\|L-R\|_1\le\varepsilon:=2^{-\tau/2}$.

We use the following observation twice.  Suppose that $T$ is a probability
law and $R(\omega)=v_\omega T(\omega)$ with $v_\omega\in[0,1]$.  All sums
below are over points with $T(\omega)>0$.  Since the
full distribution induced by the hash dominates $L$ coordinatewise and $f$
is nonincreasing,
\begin{equation}
 D_f(P_{\varphi_S(X)YS}\|T)
 \le\sum_{\omega:T(\omega)>0}T(\omega)
 f\!\left(\frac{L(\omega)}{T(\omega)}\right).
 \label{eq:monotone-light}
\end{equation}
Put $u_\omega=L(\omega)/T(\omega)$.  Only coordinates with
$u_\omega<v_\omega$ can increase the right side relative to
$\sum_\omega T(\omega)f(v_\omega)$.  For any $\eta>0$, uniform continuity
on $[0,1]$ and Markov's inequality give
\begin{align}
 \sum_\omega T(\omega)[f(u_\omega)-f(v_\omega)]
 &\le\omega_f(\eta)+\frac{f(0)}{\eta}
 \sum_\omega T(\omega)|u_\omega-v_\omega|\\
 &\le\omega_f(\eta)+\frac{f(0)\varepsilon}{\eta}.
\end{align}
For $T=U_Z\times P_Y\times P_S$, one has
$v_\omega=Q_\tau(y)$, which proves \eqref{eq:light-achievability} after
choosing $\eta=\sqrt\varepsilon=2^{-\tau/4}$.  If
$\overline Q_\tau>0$, then
$R=\overline Q_\tau(U_Z\times Q_{YS}^\circ)$ by
\eqref{eq:one-shot-tail-tilted-reference}.  Applying the same observation
with this probability law $T$ and the constant
$v_\omega=\overline Q_\tau$ proves \eqref{eq:light-achievability-down}.
\end{proof}

\section{Conditional Water Filling}
\label{sec:waterfilling}

This section solves the capped optimization underlying the one-shot bounds
and shows that its optimizer is universal over the admissible generators.
The term ``water filling'' refers to the form of the optimizer in
\eqref{eq:waterfill}: starting from $t=1$, one raises the common multiplier
$t$, so that each coordinate $tp_x$ grows until it reaches the ceiling $c$;
coordinates at the ceiling remain there while the others continue to fill.
The value of $t$ is the level at which the total mass becomes one.

\begin{lemma}
\label{lem:waterfilling}
Let $p$ be a probability vector on its support and let $a>0$.  If
$c|\supp p|\ge1$, a common optimizer of \eqref{eq:gdef} and
\eqref{eq:scaled-gdef} is
\begin{equation}
 q_x^*=\min\{c,tp_x\},
 \label{eq:waterfill}
\end{equation}
where $t\ge1$ is chosen so that
$\sum_x\min\{c,tp_x\}=1$.  If $c|\supp p|<1$, an optimizer is
$q_x^*=c$ on the support, with the remaining mass assigned to dummy symbols.
In either case the displayed optimizer is independent of both $f$ and $a$.
\end{lemma}

\begin{proof}
The perspective in \eqref{eq:df-perspective} is convex in $q$.  Moving mass
from a dummy symbol to a support coordinate that is below the cap cannot
increase the objective.  Indeed, for $u>0$ and any subgradient $\xi\in\partial
f(u)$, convexity at zero gives
\begin{equation}
 f(u)-u\xi\le f(0),
 \label{eq:dummy-transfer}
\end{equation}
and $f(u)-u\xi$ is a subgradient of $q\mapsto qf(p/q)$.  Thus all support
coordinates are capped if $c|\supp p|<1$, while the mass constraint is active
if $c|\supp p|\ge1$.

In the latter case, $t\mapsto\sum_x\min\{c,tp_x\}$ is continuous and
nondecreasing, is at most one at $t=1$, and equals $c|\supp p|\ge1$ for all
sufficiently large $t$.  Hence a finite normalizing $t\ge1$ exists.

Consider the latter case and the scaled objective.  Let $u_0:=a/t$, and use
the right derivative of $f$ at nondifferentiability points.  The function
\begin{equation}
 \gamma(u):=f(u)-uf'_+(u)
\end{equation}
is nonincreasing.  Indeed, for $0<u<v$, convexity gives
\begin{align}
 \gamma(u)-\gamma(v)
 &\ge-(v-u)f'_+(v)-uf'_+(u)+vf'_+(v)\\
 &=u[f'_+(v)-f'_+(u)]\ge0.
 \label{eq:gamma-monotonicity}
\end{align}
On every uncapped coordinate in \eqref{eq:waterfill},
$ap_x/q_x^*=u_0$; on every capped coordinate,
$ap_x/q_x^*\ge u_0$.  Hence the Karush--Kuhn--Tucker conditions hold with the
mass multiplier $-\gamma(u_0)$ and cap multipliers
$\gamma(u_0)-\gamma(ap_x/c)\ge0$.  Convexity makes these conditions
sufficient.  The scalar $t$ is determined only by normalization, proving
optimality and independence from $f$ and $a$.  Taking $a=1$ recovers
\eqref{eq:gdef}.
\end{proof}

The water-filling value admits an exact representation in terms of the
surprisal $J_p=-\log p_X(X)$ under $X\sim p$.

\begin{lemma}
\label{lem:prob-rep}
Let $c=2^{-h}$ and $c|\supp p|>1$.  Let $t$ be the scalar in
Lemma~\ref{lem:waterfilling}, put $b=h+\log t$, and define
\begin{equation}
 R_p(b):=\Prb\{J_p>b\}
 +\E[2^{J_p-b}\ind\{J_p\le b\}].
 \label{eq:Rpb}
\end{equation}
Then
\begin{align}
 1&=2^{b-h}R_p(b),\label{eq:normalization-b}\\
 g_{f,2^{-h}}(p)
 &=t\Prb\{J_p>b\}f(1/t)\\
 &\quad+\E[2^{J_p-h}f(2^{h-J_p})\ind\{J_p\le b\}].
 \label{eq:g-prob-rep}
 \\
 g_{f,2^{-h}}^{(a)}(p)
 &=t\Prb\{J_p>b\}f(a/t)\\
 &\quad+\E[2^{J_p-h}f(a2^{h-J_p})\ind\{J_p\le b\}].
 \label{eq:scaled-g-prob-rep}
\end{align}
\end{lemma}

\begin{proof}
A coordinate is capped exactly when $tp_x\ge2^{-h}$, equivalently
$J_p(x)\le b$.  On the uncapped set the total $q$-mass is
$t\Prb\{J_p>b\}$; on the capped set it is
$t\E[2^{J_p-b}\ind\{J_p\le b\}]$.  Their sum is one, proving
\eqref{eq:normalization-b}.  The likelihood ratio $p_x/q_x^*$ equals $1/t$
on the uncapped set and $2^{h-J_p(x)}$ on the capped set.  Scaling it by $a$
and substituting into \eqref{eq:df-perspective} and
\eqref{eq:scaled-gdef} proves \eqref{eq:g-prob-rep} and
\eqref{eq:scaled-g-prob-rep}.
\end{proof}

The next lemma is the asymptotic link between capped divergences and
transformed surprisal-tail probabilities.  Its local uniformity in the scale
will be used when the reference marginal is optimized.

\begin{lemma}
\label{lem:f-tail}
\label{lem:scaled-f-tail}
Let $p_n$ be probability vectors, let $J_n=-\log p_n(X_n)$ under
$X_n\sim p_n$, and let $h_n$ be thresholds.  Suppose
\begin{align}
 \Prb\{J_n\ge h_n\}&\to q\in(0,1),\label{eq:tail-assumption}\\
 \sup_{|u|\le K}\Prb\{|J_n-h_n-u|\le B\}&\to0
 \quad\text{for every fixed }K,B>0,
 \label{eq:anti-concentration}
\end{align}
and $2^{-h_n}|\supp p_n|>1$ eventually.  Then, for every
$0<a_-<a_+<\infty$,
\begin{equation}
 \sup_{a\in[a_-,a_+]}
 \left|g_{f,2^{-h_n}}^{(a)}(p_n)-f(aq)\right|\to0.
 \label{eq:scaled-f-tail-limit}
\end{equation}
In particular,
\begin{equation}
 g_{f,2^{-h_n}}(p_n)\to f(q).
 \label{eq:f-tail-limit}
\end{equation}
For every $a\ge0$, one also has the bounds
\begin{equation}
 f(a)\le g_{f,c}^{(a)}(p)\le f(0).
 \label{eq:small-a-bounds}
\end{equation}
\end{lemma}

\begin{proof}
Let $b_n$ and $t_n$ be the water level and multiplier from
Lemma~\ref{lem:prob-rep}.  Since $t_n\ge1$, $b_n\ge h_n$.  Define
$F_n(b)=2^{b-h_n}R_{p_n}(b)$, so $F_n(b_n)=1$.  Choose $D>0$ with
$2^Dq>1$.  The anti-concentration condition gives
$\Prb\{J_n>h_n+D\}\to q$, and hence $F_n(h_n+D)>1$ eventually.  Since $F_n$
is nondecreasing,
\begin{equation}
 0\le b_n-h_n\le D
 \label{eq:water-level-bound}
\end{equation}
eventually, and uniform anti-concentration gives
$\Prb\{J_n>b_n\}\to q$.

As in \eqref{eq:Rpb}, split the sub-threshold term at $b_n-B$.  The far part
is at most $2^{-B}$ and the near part is $o(1)$, uniformly under
\eqref{eq:water-level-bound}.  First letting $n\to\infty$ and then
$B\to\infty$ yields $R_{p_n}(b_n)\to q$.  Therefore
$t_n=1/R_{p_n}(b_n)\to1/q$, and
$t_n\Prb\{J_n>b_n\}\to1$.  Uniform continuity of $f$ on the resulting
compact set of arguments shows that the uncapped term in
\eqref{eq:scaled-g-prob-rep} converges uniformly on $[a_-,a_+]$ to $f(aq)$.

For the capped term, put $u_n=a2^{h_n-J_n}$.  Its integrand is
$a f(u_n)/u_n$.  On $J_n\le b_n-B$, it is bounded in absolute value by
\begin{equation}
 a_+\sup_{u\ge a_-2^{B-D}}\frac{|f(u)|}{u},
\end{equation}
which tends to zero as $B\to\infty$ by the second condition in
\eqref{eq:f-assumptions}.  On the
remaining fixed-width strip, $u_n$ stays in a compact subset of
$(0,\infty)$ uniformly in $a$, while anti-concentration makes the probability
of the strip vanish.  This proves \eqref{eq:scaled-f-tail-limit}; taking an
interval containing $a=1$ gives \eqref{eq:f-tail-limit}.

Finally, apply Jensen's inequality to the optimizer from
Lemma~\ref{lem:waterfilling}, including its dummy mass.  The weighted average
of the likelihood-ratio arguments is $a$, so the objective is at least
$f(a)$.  Since $f$ is nonincreasing, every term is at most $f(0)$, proving
\eqref{eq:small-a-bounds}.
\end{proof}

\section{Conditional Product Asymptotics}
\label{sec:conditional-asymptotics}

This section derives the conditional and unconditional Gaussian limits used
to evaluate the one-shot bounds.

We begin with the required normal approximation.  The displayed form is the
classical Berry--Esseen inequality for independent, not necessarily
identically distributed summands; see, e.g., \cite{chen_shao_2001}.

\begin{fact}[Berry--Esseen]
\label{fact:BE}
If $W_1,\ldots,W_n$ are independent, centered, finite-valued random variables
and $B_n^2=\sum_i\E W_i^2>0$, then for a universal constant
$C_{\mathrm{BE}}$,
\begin{equation}
 \sup_t\left|\Prb\!\left\{\frac{\sum_iW_i}{B_n}\le t\right\}
 -\Phi(t)\right|
 \le C_{\mathrm{BE}}\frac{\sum_i\E|W_i|^3}{B_n^3}.
 \label{eq:BE}
\end{equation}
\end{fact}

For $y$ with $P_Y(y)>0$, abbreviate
\begin{align}
 h(y)&:=\E_{X|Y=y}[\imath_{X|Y}(X|y)]=H(X|Y=y),\label{eq:hy}\\
 v(y)&:=\Var_{X|Y=y}[\imath_{X|Y}(X|y)],\label{eq:ve}\\
 \rho(y)&:=\E_{X|Y=y}|\imath_{X|Y}(X|y)-h(y)|^3.
 \label{eq:rhoe}
\end{align}
Finite alphabets imply $\rho_*:=\max_y\rho(y)<\infty$.  Define
the block conditional surprisal and the threshold parametrization by
\begin{align}
 \imath_n(X^n|y^n)&:=-\log P_{X^n|Y^n}(X^n|y^n),
 \label{eq:block-conditional-surprisal}\\
 h_n(x)&:=nH-\sqrt{nV}\,x.
 \label{eq:hn}
\end{align}
Conditional on $Y^n=y^n$, the mean and variance of the block surprisal are
\begin{align}
 \mu_n(y^n)&:=\sum_{i=1}^nh(y_i),
 \label{eq:conditional-mean}\\
 \sigma_n^2(y^n)&:=\sum_{i=1}^nv(y_i).
 \label{eq:conditional-variance}
\end{align}
Finally, define
\begin{equation}
 Q_n(Y^n;x):=\Prb\{\imath_n(X^n|Y^n)\ge h_n(x)|Y^n\}.
 \label{eq:Qn}
\end{equation}

\begin{lemma}
\label{lem:conditional-CLT}
Let $Z_n(Y^n) := \frac{\mu_n(Y^n) - nH}{\sqrt{n}}$.
If $V_2>0$, then for every fixed $x\in\mathbb R$,
\begin{align}
 Q_n(Y^n;x)-
 \Phi\!\left(\frac{x\sqrt V+Z_n(Y^n)}{\sqrt{V_2}}\right)
 &\longrightarrow0
 \quad\text{in probability},
 \label{eq:conditional-CLT}\\
 Z_n(Y^n)&\Longrightarrow Z\sim N(0,V_1).
 \label{eq:center-CLT}
\end{align}
\end{lemma}

\begin{proof}
Conditional on $Y^n=y^n$, the $n$ surprisal summands are independent, with
mean sum $\mu_n(y^n)$, variance
$\sigma_n^2(y^n)=\sum_i v(y_i)$, and third absolute centered moments
$\rho(y_i)$.  The law of large numbers gives
$\sigma_n^2(Y^n)/n\to V_2$ in probability.  On
$\sigma_n^2(y^n)\ge nV_2/2$, Fact~\ref{fact:BE} gives a uniform error at most
\begin{equation}
 \frac{2^{3/2}C_{\mathrm{BE}}\rho_*}{V_2^{3/2}\sqrt n}.
 \label{eq:conditional-BE-rate}
\end{equation}
Furthermore,
\begin{equation}
 \frac{\mu_n(y^n)-h_n(x)}{\sigma_n(y^n)}
 =\frac{\sqrt n\,[x\sqrt V+Z_n(y^n)]}{\sigma_n(y^n)}.
\end{equation}
Combining these observations proves \eqref{eq:conditional-CLT}.  Finally,
$Z_n=n^{-1/2}\sum_i[h(Y_i)-H]$; the ordinary central limit theorem gives
\eqref{eq:center-CLT}, with the degenerate interpretation $Z_n=0$ if $V_1=0$.
\end{proof}

The next lemma transfers the conditional tail approximation to capped
divergences, locally uniformly in the scale used for reference-marginal
optimization.  Let
\begin{equation}
 d(y):=\log|\supp P_{X|Y=y}|-h(y)\ge0,
 \label{eq:entropy-defect}
\end{equation}
and put $D_0:=\E d(Y)$.  Equality in \eqref{eq:entropy-defect} holds exactly
for a uniform conditional fibre, so $V_2>0$ implies $D_0>0$.  For $K>0$ let
$\mathcal T_{n,K}$ be the event on which
\begin{align}
 |Z_n|&\le K,\qquad
 \left|\frac{\sigma_n^2}{n}-V_2\right|\le\frac{V_2}{2},\\
 \left|\frac1n\sum_{i=1}^nd(Y_i)-D_0\right|&\le\frac{D_0}{2}.
 \label{eq:typical-event}
\end{align}

\begin{lemma}
\label{lem:conditional-capped}
\label{lem:uniform-conditional-scaled}
Assume $V_2>0$ and fix $x\in\mathbb R$ and $K>0$.  For every
$0<a_-<a_+<\infty$,
\begin{equation}
 \sup_{y^n\in\mathcal T_{n,K}}\sup_{a\in[a_-,a_+]}
 \left|g_{f,2^{-h_n(x)}}^{(a)}(P_{X^n|Y^n=y^n})
 -f(aQ_n(y^n;x))\right|\longrightarrow0.
 \label{eq:uniform-conditional-scaled}
\end{equation}
Consequently,
\begin{equation}
 \E_{Y^n}\left|
 g_{f,2^{-h_n(x)}}(P_{X^n|Y^n})-f(Q_n(Y^n;x))
 \right|\to0.
 \label{eq:conditional-capped-limit}
\end{equation}
\end{lemma}

\begin{proof}
On this event the Gaussian argument in the Berry--Esseen approximation used
in the proof of \eqref{eq:conditional-CLT} remains in a compact interval.
Thus, by \eqref{eq:conditional-BE-rate}, for large $n$, there is $a_K>0$ such that
\begin{equation}
 a_K\le Q_n(y^n;x)\le1-a_K
 \quad\text{for every }y^n\in\mathcal T_{n,K}.
 \label{eq:q-bounded-away}
\end{equation}
Fact~\ref{fact:BE}, applied at the two endpoints of any fixed-width interval,
also gives, for fixed $B,T>0$,
\begin{equation}
 \sup_{y^n\in\mathcal T_{n,K}}\sup_{|u|\le T}
 \Prb\{|\imath_n-h_n(x)-u|\le B|y^n\}=O(n^{-1/2}).
 \label{eq:uniform-anticoncentration}
\end{equation}

The conditional support obeys
 \begin{align}
 &\log|\supp P_{X^n|Y^n=y^n}|-h_n(x)\\
 &\quad=\sum_{i=1}^nd(y_i)+\mu_n(y^n)-h_n(x).
 \label{eq:support-gap}
\end{align}
On $\mathcal T_{n,K}$, the first term is at least $nD_0/2$ and the second is
$O_K(\sqrt n)$.  Hence the support condition of Lemma~\ref{lem:f-tail} holds
uniformly on $\mathcal T_{n,K}$.  

Suppose that \eqref{eq:uniform-conditional-scaled} were false.  There would
then exist $\epsilon_0>0$, a subsequence $n_j$, points
$y^{n_j}\in\mathcal T_{n_j,K}$, and scales $a_j\in[a_-,a_+]$ for which the
displayed absolute difference is at least $\epsilon_0$.  Passing to a further
subsequence, compactness and \eqref{eq:q-bounded-away} give
\begin{equation}
 a_j\to a\in[a_-,a_+],
 \qquad Q_{n_j}(y^{n_j};x)\to q\in[a_K,1-a_K].
\end{equation}
For the sequence of conditional probability vectors
$P_{X^{n_j}|Y^{n_j}=y^{n_j}}$, the tail assumption of
Lemma~\ref{lem:scaled-f-tail} holds with limit $q$, the anti-concentration
assumption follows from \eqref{eq:uniform-anticoncentration}, and the support
condition follows from \eqref{eq:support-gap}.  The same lemma therefore gives,
uniformly for $a'\in[a_-,a_+]$,
\begin{equation}
 g_{f,2^{-h_{n_j}(x)}}^{(a')}
 (P_{X^{n_j}|Y^{n_j}=y^{n_j}})\longrightarrow f(a'q).
\end{equation}
Evaluating at $a'=a_j$ and using continuity of $f$ on the compact interval
$[a_-a_K,a_+]$ makes both terms in
\eqref{eq:uniform-conditional-scaled} converge to $f(aq)$, a contradiction.

Taking any scale interval containing $a=1$ proves uniform convergence of the
unscaled difference on $\mathcal T_{n,K}$.  Both unscaled terms lie in
$[0,f(0)]$ by \eqref{eq:small-a-bounds} and monotonicity.
The laws of large numbers and \eqref{eq:center-CLT} imply
$\lim_{K\to\infty}\limsup_n\Prb\{\mathcal T_{n,K}^c\}=0$.  For fixed $K$,
the expectation in \eqref{eq:conditional-capped-limit} is bounded by the
uniform error on $\mathcal T_{n,K}$ plus
$2f(0)\Prb\{\mathcal T_{n,K}^c\}$.  First taking $n\to\infty$ and then
$K\to\infty$ proves the claim.
\end{proof}

The following lemma isolates the perspective minimization that turns
fibrewise scaled limits into an optimized-marginal limit.

\begin{lemma}
\label{lem:perspective-aggregation}
For each $n$, let $\pi_n$ be a probability distribution on a finite set
$\mathcal E_n$, let $q_n:\mathcal E_n\to[0,1]$, and let
$G_n:\mathcal E_n\times[0,\infty)\to\mathbb R$.  Suppose that
\begin{equation}
 \bar q_n:=\sum_e\pi_n(e)q_n(e)\longrightarrow q_0\in(0,1),
 \label{eq:aggregation-average}
\end{equation}
and, for every $e$ and $a\ge0$,
\begin{equation}
 f(a)\le G_n(e,a)\le f(0).
 \label{eq:aggregation-bounds}
\end{equation}
Assume that there are sets $\mathcal T_{n,K}\subseteq\mathcal E_n$ such that
\begin{equation}
 \lim_{K\to\infty}\limsup_{n\to\infty}
 \pi_n(\mathcal T_{n,K}^c)=0,
 \label{eq:aggregation-tightness}
\end{equation}
and, for each fixed $K$, some $\eta_K\in(0,1/2)$ satisfies
\begin{equation}
 \eta_K\le q_n(e)\le1-\eta_K
 \quad(e\in\mathcal T_{n,K})
 \label{eq:aggregation-tail-bounds}
\end{equation}
eventually.  Finally, suppose that for every fixed $K$ and
$0<a_-<a_+<\infty$,
\begin{equation}
 \sup_{e\in\mathcal T_{n,K}}\sup_{a\in[a_-,a_+]}
 |G_n(e,a)-f(aq_n(e))|\longrightarrow0.
 \label{eq:aggregation-local-limit}
\end{equation}
Then
\begin{equation}
 \inf_{R_n}\sum_{e:R_n(e)>0}R_n(e)
 G_n\!\left(e,\frac{\pi_n(e)}{R_n(e)}\right)
 \longrightarrow f(q_0),
 \label{eq:perspective-aggregation-limit}
\end{equation}
where the infimum is over probability distributions on $\mathcal E_n$ and
reference-zero terms use the same convention as \eqref{eq:f-divergence}.
\end{lemma}

\begin{proof}
Fix $K$ and abbreviate $\mathcal T_n:=\mathcal T_{n,K}$ and
$\eta:=\eta_K$.  For a reference $R_n$, write
$a(e):=\pi_n(e)/R_n(e)$ where $R_n(e)>0$, and define
\begin{equation}
 r_n(e):=q_n(e)\ind\{e\in\mathcal T_n\}
 +\ind\{e\notin\mathcal T_n\}.
 \label{eq:aggregation-modified-tail}
\end{equation}
Let $0<\varepsilon<1$ and choose $A>1/\eta$.  On $\mathcal T_n$, if
$a(e)<\varepsilon$, then \eqref{eq:aggregation-bounds} and monotonicity give
\begin{equation}
 G_n(e,a(e))\ge f(a(e))
 \ge f(a(e)q_n(e))-\omega_f(\varepsilon).
 \label{eq:aggregation-small-scale}
\end{equation}
For $\varepsilon\le a(e)\le A$, let the uniform error in
\eqref{eq:aggregation-local-limit} be
$\delta_{n,K}(\varepsilon,A)$.  For $a(e)>A$, put
\begin{equation}
 \rho_f(A,\eta):=
 \sup_{\substack{a>A\\q\in[\eta,1]}}
 \frac{f(aq)-f(a)}{a}.
 \label{eq:aggregation-large-scale-rho}
\end{equation}
Because $A\eta>1$, monotonicity and the second condition in
\eqref{eq:f-assumptions} give
\begin{equation}
 0\le\rho_f(A,\eta)
 \le\sup_{u>A}\frac{|f(u)|}{u}
 +\sup_{u>A\eta}\frac{|f(u)|}{u}
 \longrightarrow0.
 \label{eq:aggregation-large-scale-limit}
\end{equation}
Thus \eqref{eq:aggregation-bounds} implies
 \begin{align}
 &R_n(e)G_n(e,a(e))\\
 &\quad\ge R_n(e)f(a(e)q_n(e))
 -\pi_n(e)\rho_f(A,\eta)
 \label{eq:aggregation-large-scale}
\end{align}
on the large-scale region.  Outside $\mathcal T_n$,
\eqref{eq:aggregation-bounds} gives
$G_n(e,a(e))\ge f(a(e))=f(a(e)r_n(e))$.
Summing the four regions yields
 \begin{align}
 &\sum_{e:R_n(e)>0}R_n(e)G_n(e,a(e))\\
 &\quad\ge\sum_{e:R_n(e)>0}R_n(e)f(a(e)r_n(e))
 -\omega_f(\varepsilon)-\delta_{n,K}(\varepsilon,A)
 -\rho_f(A,\eta).
 \label{eq:aggregation-uniform-lower}
\end{align}
Jensen's inequality and monotonicity of $f$ give
\begin{align}
 \sum_{e:R_n(e)>0}R_n(e)f(a(e)r_n(e))
 &\ge f\!\left(\sum_{e:R_n(e)>0}\pi_n(e)r_n(e)\right)\\
 &\ge f\!\left(\sum_e\pi_n(e)r_n(e)\right).
 \label{eq:aggregation-jensen}
\end{align}
The last step also handles reference-zero points: omitting them decreases the
argument of the nonincreasing function $f$.  Moreover,
\begin{equation}
 0\le\sum_e\pi_n(e)r_n(e)-\bar q_n
 \le\pi_n(\mathcal T_n^c).
 \label{eq:aggregation-tail-error}
\end{equation}
Taking the infimum over $R_n$ and then sending successively
$n\to\infty$, $\varepsilon\downarrow0$, $A\to\infty$, and $K\to\infty$
proves the lower limit in \eqref{eq:perspective-aggregation-limit}.

For the upper limit, choose
\begin{equation}
 R_n^\circ(e):=\frac{\pi_n(e)q_n(e)}{\bar q_n}
 \label{eq:aggregation-tilted-reference}
\end{equation}
on points where $q_n(e)>0$.  On $\mathcal T_n$, its scale satisfies
$a(e)q_n(e)=\bar q_n$ and belongs to a compact subinterval of
$(0,\infty)$.  Hence \eqref{eq:aggregation-local-limit} and
\eqref{eq:aggregation-bounds} give
 \begin{align}
 &\sum_{e:R_n^\circ(e)>0}R_n^\circ(e)G_n(e,a(e))\\
 &\quad\le f(\bar q_n)+o(1)
 +\frac{f(0)}{\bar q_n}\pi_n(\mathcal T_n^c).
 \label{eq:aggregation-upper}
\end{align}
First letting $n\to\infty$ and then $K\to\infty$ proves the matching upper
limit.
\end{proof}

The endpoint $V_2=0$ needs a separate elementary argument.

\begin{lemma}
\label{lem:uniform-endpoint}
Assume $V_2=0$ and $V>0$.  Every nonzero $P_{X|Y=y}$ is uniform on its
support, and
\begin{equation}
 \E_{Y^n}g_{f,2^{-h_n(x)}}(P_{X^n|Y^n})
 \to f(0)\Phi(-x).
 \label{eq:uniform-g-limit}
\end{equation}
If $\tau_n\to\infty$ and $\tau_n=o(\sqrt n)$, then
\begin{equation}
 \E_{Y^n}f\!\left(
 \Prb\{\imath_n\ge h_n(x)+\tau_n|Y^n\}
 \right)
 \to f(0)\Phi(-x).
 \label{eq:uniform-tail-limit}
\end{equation}
\end{lemma}

\begin{proof}
Since $v(Y)\ge0$ and $\E v(Y)=0$, $v(Y)=0$ almost surely.  Thus the
conditional surprisal is constant on every nonzero fibre.  If that fibre has
size $N_y$, its probabilities are $1/N_y$ and $h(y)=\log N_y$.
Conditional on $y^n$, the product law is uniform on
$N(y^n)=2^{\mu_n(y^n)}$ points.  Put
$a_n(y^n):=2^{\mu_n(y^n)-h_n(x)}$.  Directly from
\eqref{eq:df-perspective} and Lemma~\ref{lem:waterfilling},
\begin{equation}
 g_{f,2^{-h_n(x)}}(P_{X^n|Y^n=y^n})
 =\begin{cases}
 0,&a_n(y^n)\ge1,\\
 a_n(y^n)f(1/a_n(y^n))+(1-a_n(y^n))f(0),&a_n(y^n)<1.
 \end{cases}
 \label{eq:uniform-exact-g}
\end{equation}
Here $V=V_1>0$, and
\begin{equation}
 \frac{\mu_n(Y^n)-h_n(x)}{\sqrt n}
 \Longrightarrow Z+x\sqrt V,
 \qquad Z\sim N(0,V).
 \label{eq:uniform-center-clt}
\end{equation}
When the left side is bounded above by $-\eta$, the second line of
\eqref{eq:uniform-exact-g} converges uniformly to $f(0)$ because
$f(t)/t\to0$; when it is at least $\eta$, the value is zero.  The expression
is bounded by $f(0)$.  Sandwiching, then letting $n\to\infty$ and
$\eta\downarrow0$, proves \eqref{eq:uniform-g-limit}.

The conditional tail in \eqref{eq:uniform-tail-limit} is the zero-one
variable $\ind\{\mu_n(Y^n)\ge h_n(x)+\tau_n\}$.  Since $f(1)=0$ and
$\tau_n/\sqrt n\to0$, \eqref{eq:uniform-center-clt} proves the second claim.
\end{proof}

The preceding conditional limit results now identify the same
fixed-reference Gaussian profile in two forms: both the capped one-shot
quantity used in the converse and the conditional light-tail quantity used
in achievability converge to it.

\begin{proposition}
\label{prop:Gaussian-profile}
For every fixed $x\in\mathbb R$ and $V>0$,
\begin{align}
 \Gamma_{f,P_{XY}^{\times n}}^\uparrow(h_n(x))
 &\to\Fprof_{f,V_1/V}(x),
 \label{eq:Gamma-limit}\\
 \E_{Y^n}f\!\left(
 \Prb\{\imath_n\ge h_n(x)+\tau_n|Y^n\}
 \right)
 &\to\Fprof_{f,V_1/V}(x)
 \label{eq:tail-profile-limit}
\end{align}
for every $\tau_n\to\infty$ with $\tau_n=o(\sqrt n)$.
\end{proposition}

\begin{proof}
Suppose $V_2>0$.  Lemmas~\ref{lem:conditional-CLT} and
\ref{lem:conditional-capped}, boundedness of $f$ on $[0,1]$, and the
continuous-mapping theorem give
\begin{align}
 \Gamma_{f,P^{\times n}}^\uparrow(h_n(x))
 &\to\E_{Z\sim N(0,V_1)}
 f\!\left(\Phi\!\left(\frac{x\sqrt V+Z}{\sqrt{V_2}}\right)\right)\\
 &=\Fprof_{f,V_1/V}(x).
 \label{eq:profile-change-variables}
\end{align}
For the shifted tail, write
\begin{equation}
 Q_n^{(\tau)}(Y^n;x):=\Prb\{\imath_n\ge h_n(x)+\tau_n\mid Y^n\}.
\end{equation}
On $\{\sigma_n^2\ge nV_2/2\}$, Fact~\ref{fact:BE} and
\eqref{eq:conditional-BE-rate} give
\begin{equation}
 Q_n^{(\tau)}(Y^n;x)-
 \Phi\!\left(
 \frac{x\sqrt V+Z_n(Y^n)-\tau_n/\sqrt n}{\sigma_n(Y^n)/\sqrt n}
 \right)\longrightarrow0
 \quad\text{in probability}.
 \label{eq:shifted-conditional-CLT}
\end{equation}
Because $\tau_n/\sqrt n\to0$ and
$\sigma_n(Y^n)/\sqrt n\to\sqrt{V_2}$ in probability, Slutsky's theorem and
\eqref{eq:center-CLT} show that the argument of $f$ in
\eqref{eq:shifted-conditional-CLT} converges in distribution to the Gaussian
tail in \eqref{eq:profile-change-variables}.  The random variables
$f(Q_n^{(\tau)})$ are bounded by $f(0)$, so convergence in distribution plus
boundedness and continuity of $f$ gives convergence of their expectations,
proving \eqref{eq:tail-profile-limit}.  If $V_2=0$,
Lemma~\ref{lem:uniform-endpoint} gives both conclusions and
\eqref{eq:F-endpoints} identifies their common limit.
\end{proof}

Optimizing the reference marginal collapses this conditional profile to $f$
applied to a single unconditional Gaussian tail.  The next proposition gives
this optimized limit together with the unconditional tail limit used in the
achievability argument.

\begin{proposition}
\label{prop:Gaussian-profile-down}
For every fixed $x\in\mathbb R$ and $V>0$,
\begin{equation}
 \Gamma_{f,P_{XY}^{\times n}}^\downarrow(h_n(x))
 \to f(\Phi(x)).
 \label{eq:Gamma-down-limit}
\end{equation}
Moreover, if $\tau_n\to\infty$ and $\tau_n=o(\sqrt n)$, then
\begin{equation}
 \overline Q_n(x):=\Prb\{\imath_n(X^n|Y^n)
 \ge h_n(x)+\tau_n\}\to\Phi(x).
 \label{eq:unconditional-tail-limit}
\end{equation}
\end{proposition}

\begin{proof}
Let $I_i:=\imath_{X|Y}(X_i|Y_i)$.  These random variables are i.i.d., with
mean $H$ and variance $V$.  Therefore
\begin{align}
 \Prb\!\left\{\sum_{i=1}^nI_i\ge h_n(x)+\tau_n\right\}
 &=\Prb\!\left\{
 \frac{\sum_iI_i-nH}{\sqrt{nV}}
 \ge-x+\frac{\tau_n}{\sqrt{nV}}\right\}\\
 &\longrightarrow\Phi(x),
 \label{eq:unconditional-clt-detail}
\end{align}
which proves \eqref{eq:unconditional-tail-limit}; taking $\tau_n=0$ gives
the corresponding unshifted limit.  In the remainder write
\begin{equation}
 q_n(y^n):=Q_n(y^n;x),\qquad
 \bar q_n:=\E q_n(Y^n)\longrightarrow q_0:=\Phi(x)\in(0,1).
 \label{eq:qn-bar-unshifted}
\end{equation}

When $V_2>0$, apply Lemma~\ref{lem:perspective-aggregation} with
\begin{equation}
 \mathcal E_n=\supp P_Y^{\times n},\quad
 \pi_n=P_Y^{\times n},\quad
 G_n(y^n,a)=g_{f,2^{-h_n(x)}}^{(a)}(P_{X^n|Y^n=y^n}).
 \label{eq:aggregation-specialization}
\end{equation}
The bounds \eqref{eq:aggregation-bounds} are
\eqref{eq:small-a-bounds}; the local limit and tail bounds are
\eqref{eq:uniform-conditional-scaled} and \eqref{eq:q-bounded-away}; and
\eqref{eq:aggregation-tightness} follows from the laws of large numbers and
\eqref{eq:center-CLT}.  Thus
\begin{equation}
 \Gamma_{f,P_{XY}^{\times n}}^\downarrow(h_n(x))\to f(q_0)
 \label{eq:optimized-profile-positive-v2}
\end{equation}
whenever $V_2>0$.

It remains to treat $V_2=0$.  Every conditional product distribution is then
uniform on $N(y^n)=2^{\mu_n(y^n)}$ points.  Put
\begin{equation}
 s_n(y^n):=2^{-h_n(x)}N(y^n)=2^{\mu_n(y^n)-h_n(x)}.
 \label{eq:uniform-scale-sn}
\end{equation}
The scaled water-filling value is exactly
\begin{equation}
 g_{f,2^{-h_n(x)}}^{(a)}(P_{X^n|Y^n=y^n})
 =\begin{cases}
 f(a),&s_n(y^n)\ge1,\\
 s_n(y^n)f\!\left(\dfrac{a}{s_n(y^n)}\right)
 +(1-s_n(y^n))f(0),&s_n(y^n)<1.
 \end{cases}
 \label{eq:uniform-scaled-g}
\end{equation}
Choose any $\kappa_n\to\infty$ with $\kappa_n=o(\sqrt n)$ and define
\begin{equation}
 \mathcal C_n:=\{y^n:s_n(y^n)\ge2^{-\kappa_n}\},
 \qquad \mathcal B_n:=\mathcal C_n^c.
 \label{eq:uniform-CB-sets}
\end{equation}
Since $V=V_1>0$, \eqref{eq:uniform-center-clt} gives
\begin{align}
 P_Y^{\times n}(\mathcal C_n)&\to q_0,
 \label{eq:uniform-C-probability}\\
 P_Y^{\times n}\{2^{-\kappa_n}\le s_n(Y^n)<1\}&\to0.
 \label{eq:uniform-strip-probability}
\end{align}

For an arbitrary reference $Q_{Y^n}$, put
$a(y^n):=P_Y^{\times n}(y^n)/Q_{Y^n}(y^n)$ whenever the denominator is
positive.  Equation \eqref{eq:small-a-bounds} yields
$g^{(a)}_{f,c}\ge f(a)$ on $\mathcal C_n$.  On $\mathcal B_n$, put
$R_n(y^n):=Q_{Y^n}(y^n)s_n(y^n)$.  The loss relative to $f(0)$ is
\begin{align}
 &\sum_{y^n\in\mathcal B_n}Q_{Y^n}(y^n)
 \left[f(0)-g_{f,2^{-h_n(x)}}^{(a(y^n))}
 (P_{X^n|Y^n=y^n})\right]\\
 &\quad=\sum_{\substack{y^n\in\mathcal B_n\\R_n(y^n)>0}}
 R_n(y^n)\left[f(0)-f\!\left(
 \frac{P_Y^{\times n}(y^n)}{R_n(y^n)}\right)\right].
 \label{eq:uniform-B-loss}
\end{align}
This loss tends to zero uniformly over $Q_{Y^n}$.  Indeed,
$\sum_{\mathcal B_n}R_n\le2^{-\kappa_n}$.  For any fixed $T>1$, the terms
whose displayed likelihood ratio is at most $T$ are bounded by
\begin{equation}
 2^{-\kappa_n}\max_{0\le u\le T}[f(0)-f(u)],
\end{equation}
whereas the remaining terms are at most
\begin{equation}
 \sup_{u>T}\frac{f(0)-f(u)}{u}
 \sum_{y^n\in\mathcal B_n}P_Y^{\times n}(y^n)
 \le\sup_{u>T}\frac{f(0)-f(u)}{u}.
\end{equation}
The latter supremum tends to zero as $T\to\infty$ by the second condition in
\eqref{eq:f-assumptions}.  More explicitly, define
\begin{equation}
 \zeta_n:=\inf_{T>1}\left\{
 2^{-\kappa_n}\max_{0\le u\le T}[f(0)-f(u)]
 +\sup_{u>T}\frac{f(0)-f(u)}{u}\right\}.
 \label{eq:zeta-uniform-endpoint}
\end{equation}
For every $\epsilon>0$, first choose a fixed $T$ making the second term less
than $\epsilon$, and then take $n$ large enough that the first term is less
than $\epsilon$.  Hence $\zeta_n\to0$, uniformly over the reference law, and
\begin{align}
 \Gamma_{f,P_{XY}^{\times n}}^\downarrow(h_n(x))
 &\ge\inf_{Q_{Y^n}}\left\{
 \sum_{y^n\in\mathcal C_n}Q_{Y^n}(y^n)f(a(y^n))
 +Q_{Y^n}(\mathcal B_n)f(0)\right\}-\zeta_n\\
 &\ge f(P_Y^{\times n}(\mathcal C_n))-\zeta_n.
 \label{eq:uniform-endpoint-lower-full}
\end{align}
The final line is Jensen's inequality applied to arguments $a(y^n)$ on
$\mathcal C_n$ and zero on $\mathcal B_n$; reference-zero points are handled
as in \eqref{eq:aggregation-jensen}.

For the upper bound, take
\begin{equation}
 Q_{Y^n}^{\circ}(y^n)
 :=\frac{P_Y^{\times n}(y^n)}{P_Y^{\times n}(\mathcal C_n)}
 \ind\{y^n\in\mathcal C_n\}.
 \label{eq:uniform-endpoint-reference}
\end{equation}
Its scale is the constant $P_Y^{\times n}(\mathcal C_n)$.  The first line of
\eqref{eq:uniform-scaled-g} therefore contributes exactly
$f(P_Y^{\times n}(\mathcal C_n))$ on $\{s_n\ge1\}$; on the shrinking strip
$\{2^{-\kappa_n}\le s_n<1\}$ the value is at most $f(0)$.  Hence
\begin{align}
 \Gamma_{f,P_{XY}^{\times n}}^\downarrow(h_n(x))
 &\le f(P_Y^{\times n}(\mathcal C_n))\\
 &\quad+\frac{f(0)}{P_Y^{\times n}(\mathcal C_n)}
 P_Y^{\times n}\{2^{-\kappa_n}\le s_n(Y^n)<1\}.
 \label{eq:uniform-endpoint-upper-full}
\end{align}
Equations \eqref{eq:uniform-C-probability} and
\eqref{eq:uniform-strip-probability}, together with continuity of $f$, make
the lower and upper bounds converge to $f(q_0)$.  This completes the proof in
both variance regimes.
\end{proof}

\section{Proof of the Main Results}
\label{sec:main-proof}

This section combines the one-shot bounds, water-filling formulas, and
Gaussian limits to prove the main theorem and its fixed-leakage consequence.

\begin{proof}[Proof of the fixed-marginal statement of
Theorem~\ref{thm:main-f}]
Set
\begin{equation}
 x_n:=\frac{nH-\log M_n}{\sqrt{nV}},
 \qquad x_n\to x:=-\frac{L}{\sqrt V}.
 \label{eq:xn}
\end{equation}
For every $\eta>0$, eventually $x-\eta\le x_n\le x+\eta$.  Since
$h_n(x)$ decreases with $x$, the feasible capped set expands with $x$ and
both $\Gamma_{f,P^{\times n}}^\uparrow(h_n(x))$ and the transformed conditional-tail
expectation in \eqref{eq:tail-profile-limit} are nonincreasing in $x$.
Consequently each quantity at $x_n$ lies between its values at $x-\eta$ and
$x+\eta$.  Proposition~\ref{prop:Gaussian-profile} and continuity of
$\Fprof_{f,V_1/V}$ therefore imply the same two limits with $x_n$ in place
of $x$ after $n\to\infty$ and then $\eta\downarrow0$.

\emph{Converse.}
The proof associates each output binning with a subprobability vector capped
by $1/M_n$.  The one-shot converse reduces the leakage to its optimal
$f$-divergence perspective, and conditional water filling followed by the
Gaussian-profile limit evaluates this quantity.  Formally, apply
Lemma~\ref{lem:oneshot-converse} to $P_{XY}^{\times n}$.  By
\eqref{eq:Gamma-limit}, every seeded extractor satisfies
\begin{equation}
 \liminf_{n\to\infty}
 D_f(P_{\varphi_{S_n}(X^n)Y^nS_n}\|
 U_{Z_n}\times P_Y^{\times n}\times P_{S_n})
 \ge\Fprof_{f,V_1/V}(x).
 \label{eq:main-converse}
\end{equation}
Taking the infimum gives the converse for $d_f^\uparrow$.

\emph{Achievability.}
The proof hashes only the light part of each conditional source with a
two-universal family; closeness to the corresponding scaled uniform law and
continuity of $f$ then reduce the leakage to the same conditional Gaussian
profile.  Choose $S_n$ to index a two-universal family and set
$\tau_n=n^{1/4}$.  The light condition in Lemma~\ref{lem:light-achievability}
is exactly
\begin{equation}
 \imath_n(X^n|Y^n)\ge\log M_n+\tau_n
 =h_n(x_n)+\tau_n.
\end{equation}
The modulus-of-continuity error in \eqref{eq:light-achievability} vanishes.
Equation \eqref{eq:tail-profile-limit} therefore gives
\begin{equation}
 \limsup_{n\to\infty}
 D_f(P_{\varphi_{S_n}(X^n)Y^nS_n}\|
 U_{Z_n}\times P_Y^{\times n}\times P_{S_n})
 \le\Fprof_{f,V_1/V}(x).
 \label{eq:main-achievability}
\end{equation}
The bounds coincide.
\end{proof}

\begin{proof}[Proof of the optimized-marginal statement of
Theorem~\ref{thm:main-f}]
Let $x_n$ be as in \eqref{eq:xn}, so $x_n\to x=-L/\sqrt V$.
The same cap monotonicity gives, for every fixed $\eta>0$ and all large $n$,
\begin{equation}
 \Gamma_{f,P^{\times n}}^\downarrow(h_n(x-\eta))
 \ge\Gamma_{f,P^{\times n}}^\downarrow(h_n(x_n))
 \ge\Gamma_{f,P^{\times n}}^\downarrow(h_n(x+\eta)).
 \label{eq:optimized-moving-sandwich}
\end{equation}

\emph{Converse.}
The optimized-reference converse applies water filling to each scaled
conditional source and then uses perspective aggregation to place $f$
outside the side-information average, leaving a single unconditional
Gaussian tail.  Lemma~\ref{lem:oneshot-converse-down},
Proposition~\ref{prop:Gaussian-profile-down}, and continuity of
$f\circ\Phi$, followed by $\eta\downarrow0$, give
\begin{equation}
 \liminf_{n\to\infty}d_f^\downarrow(M_n;P_{XY}^{\times n})
 \ge f(\Phi(x)).
 \label{eq:main-down-converse}
\end{equation}

\emph{Achievability.}
The proof tilts the reference marginal by the conditional light mass, so
two-universal hashing produces a constant multiple of the chosen ideal law
and the remaining quantity is the unconditional surprisal tail.  Choose a
two-universal family, set $\tau_n=n^{1/4}$, and
let
\begin{equation}
 q_n(y^n):=\Prb\{\imath_n(X^n|y^n)
 \ge\log M_n+\tau_n\mid y^n\},
 \qquad
 \bar q_n:=\E q_n(Y^n).
 \label{eq:direct-light-tail-down}
\end{equation}
When $\bar q_n>0$, denote the corresponding tail-tilted reference by
\begin{equation}
 Q_{Y^nS_n}^{\circ}(y^n,s)
 :=\frac{P_Y^{\times n}(y^n)P_{S_n}(s)q_n(y^n)}{\bar q_n}.
 \label{eq:tail-tilted-reference}
\end{equation}
For all sufficiently large $n$, $\bar q_n>0$.  Apply
\eqref{eq:light-achievability-down} to the block source.  Its reference
\eqref{eq:one-shot-tail-tilted-reference} is exactly the law in
\eqref{eq:tail-tilted-reference}, and therefore
 \begin{align}
 &D_f(P_{\varphi_{S_n}(X^n)Y^nS_n}\|
 U_{Z_n}\times Q_{Y^nS_n}^{\circ})\\
 &\qquad\le f(\bar q_n)
 +\omega_f(2^{-\tau_n/4})+f(0)2^{-\tau_n/4}.
 \label{eq:main-down-achievability-bound}
\end{align}
The unconditional tail sandwich used above, now with the shift $\tau_n$, and
\eqref{eq:unconditional-tail-limit} give $\bar q_n\to\Phi(x)$.  Taking the limsup in
\eqref{eq:main-down-achievability-bound} proves the matching upper bound.
\end{proof}

The remaining lemma records the analytic properties used for fixed leakage.

\begin{lemma}
\label{lem:profile-regularity}
For fixed $f$ and $r\in[0,1]$, the map
$x\mapsto\Fprof_{f,r}(x)$ is continuous and nonincreasing, with limits
$f(0)$ and zero at $-\infty$ and $+\infty$, respectively.  If $f$ is strictly
decreasing on $(0,1)$, the profile is strictly decreasing.  Moreover, the
value in \eqref{eq:F-endpoints} at $r=1$ equals the limit of
\eqref{eq:F-profile} as $r\uparrow1$.
\end{lemma}

\begin{proof}
For $r<1$, $Q_r(g;x)$ is continuous and strictly increasing in $x$ for every
$g\in\mathbb R$.  The assertions follow from continuity, monotonicity, and
boundedness of $f$ on $[0,1]$, together with dominated convergence.  Strict decrease follows
when $f$ is strictly decreasing.  The case $r=1$ follows from
\eqref{eq:F-endpoints}.  The calculation following
\eqref{eq:F-profile} proves the asserted continuity as $r\uparrow1$.
\end{proof}

\begin{proof}[Proof of Corollary~\ref{cor:fixed-f-leakage}]
Put $r=V_1/V$ and define, in parallel for the two marginal treatments,
\begin{align}
 \Psi^\uparrow(x)&:=\Fprof_{f,r}(x),
 &x_\uparrow&:=\Fprof_{f,r}^{-1}(\delta),
 \label{eq:fixed-leakage-profiles-up}\\
 \Psi^\downarrow(x)&:=f(\Phi(x)),
 &x_\downarrow&:=\Phi^{-1}(f^{-1}(\delta)).
 \label{eq:fixed-leakage-profiles-down}
\end{align}
Lemma~\ref{lem:profile-regularity} and the assumptions on $f$ show that both
profiles are continuous and strictly decreasing from $f(0)$ to zero, and
that $\Psi^\star(x_\star)=\delta$ for
$\star\in\{\downarrow,\uparrow\}$.

We will round the two comparison rates explicitly.  First observe that
$V>0$ implies $H>0$.  Indeed, conditional surprisal is nonnegative; if
$H=0$, then it vanishes almost surely.  Every conditional law on its support
is consequently a point mass, which forces both $V_1$ and $V_2$ to vanish,
contrary to $V>0$.

Fix $\eta>0$ and one of the arrows $\star$.  For each $n$, set
\begin{equation}
 k_{n,\pm}^\star
 :=\left\lfloor\max\left\{0,
 nH-\sqrt{nV}\,(x_\star\pm\eta)\right\}\right\rfloor.
 \label{eq:rounded-comparison-lengths}
\end{equation}
Since $H>0$, the maximum with zero is inactive for all sufficiently large
$n$.  The rounding error is less than one, and hence
\begin{equation}
 \frac{k_{n,\pm}^\star-nH}{\sqrt n}
 \longrightarrow-\sqrt V\,(x_\star\pm\eta).
 \label{eq:rounded-comparison-rates}
\end{equation}
Thus the two sequences $M_{n,\pm}^\star:=2^{k_{n,\pm}^\star}$ satisfy the
second-order-rate hypothesis of Theorem~\ref{thm:main-f}.  Applying that
theorem at the two fixed offsets gives
\begin{align}
 d_f^\star(M_{n,+}^\star;P_{XY}^{\times n})
 &\longrightarrow\Psi^\star(x_\star+\eta)<\delta,
 \label{eq:rounded-good-length}\\
 d_f^\star(M_{n,-}^\star;P_{XY}^{\times n})
 &\longrightarrow\Psi^\star(x_\star-\eta)>\delta.
 \label{eq:rounded-bad-length}
\end{align}
The first rounded length is therefore feasible and the second is infeasible
for all sufficiently large $n$.  By the output-length monotonicity established
above, every feasible length is strictly smaller than $k_{n,-}^\star$.
Together with feasibility of $k_{n,+}^\star$, this proves that the admissible
set is nonempty and bounded, that its supremum is attained, and that
\begin{equation}
 k_{n,+}^\star
 \le \ell_f^\star(\delta;P_{XY}^{\times n})
 <k_{n,-}^\star
 \label{eq:integer-length-bracket}
\end{equation}
eventually.  In particular, no convergence uniform in the output length is
being used: $\eta$ is fixed before Theorem~\ref{thm:main-f} is applied.

Divide \eqref{eq:integer-length-bracket} by $\sqrt n$ after subtracting
$nH$, and use \eqref{eq:rounded-comparison-rates}.  This yields
\begin{align}
 -\sqrt V\,(x_\star+\eta)
 &\le\liminf_{n\to\infty}
 \frac{\ell_f^\star(\delta;P_{XY}^{\times n})-nH}{\sqrt n},
 \label{eq:fixed-leakage-liminf}\\
 \limsup_{n\to\infty}
 \frac{\ell_f^\star(\delta;P_{XY}^{\times n})-nH}{\sqrt n}
 &\le-\sqrt V\,(x_\star-\eta).
 \label{eq:fixed-leakage-limsup}
\end{align}
Letting $\eta\downarrow0$ proves convergence to $-\sqrt V\,x_\star$.
Substitution of the two values in
\eqref{eq:fixed-leakage-profiles-up}--\eqref{eq:fixed-leakage-profiles-down}
gives \eqref{eq:f-quantile-expansion} and
\eqref{eq:f-quantile-down-expansion}, respectively.
\end{proof}

\section{Discussion}

The condition $\lim_{t\to\infty}f(t)/t=0$ separates the constant-leakage
second-order profile studied here from qualitatively different regimes.
Forward relative entropy, with generator $t\log t$, and superlinear
divergences such as $\chi^2$ violate this condition.  At the central-limit
rate their leakage need not approach a finite constant; in particular, relative-entropy
leakage has a $\sqrt n$ scaling \cite{hayashi_tan_2017}.  Generators with
$f(0)=\infty$, such as reverse relative entropy, also require a different
direct argument because empty output bins can have infinite cost.  The theorem
is pointwise for each fixed generator.  Approaching the forward
relative-entropy endpoint requires a separate uniform analysis.  The
divergences $D_{f_\alpha}(P\|Q)$ converge to forward relative entropy for
every fixed pair of finite-alphabet laws $P,Q$ as $\alpha\uparrow1$, but this
convergence is not uniform over laws with large likelihood ratios.  The
distinction between the two statements of Theorem~\ref{thm:main-f} is likewise
essential.

As shown in Corollary~\ref{cor:total-variation}, the chosen
total-variation generator is affine on $[0,1]$, so fixing or optimizing the
reference marginal leads to the same second-order profile.  Nonlinear
generators, including the power generators in the R\'enyi specialization,
generally retain the separate dependence on $V_1$ and $V_2$ under the fixed
criterion, but not under the optimized one.

Finite alphabets provide uniform third moments and, through
Fact~\ref{fact:BE}, the conditional anti-concentration needed for water
filling.  Extensions to countable alphabets appear possible under suitable
moment, support-growth, and uniform-integrability hypotheses.  A quantitative
version of Lemma~\ref{lem:f-tail} may also yield third-order control, but that
is separate from the exact second-order profile established here.

\paragraph*{Acknowledgments}
MB acknowledges support from the European Research Council (ERC Grant
Agreement No.~948139) and the Excellence Cluster Matter and Light for Quantum
Computing (ML4Q-2). This work was supported in part by a grant of
access to OpenAI models through the ChatGPT for Academic Researchers program.
HC acknowledges support from National Science and Technology Council (NSTC
115-2628-E-002-005, NSTC 114-2119-M-001-002, and NSTC 115-2124-M-002-014) and
Ministry of Education (NTU-115V2016-1, NTU-CC115L893705, and NTU-115L900702).
MT acknowledges support from the National Research Foundation Investigatorship
Award (NRF-NRFI10-2024-0006) and the National Research Foundation, Singapore,
through the National Quantum Office, hosted by A*STAR, under its Centre for
Quantum Technologies Funding Initiative (S24Q2d0009).

\appendices

\section{Why Fixing and Optimizing the Reference Marginal Differ}
\label{app:marginal-comparison}

The difference between the two criteria can be seen by first fixing a value
$Y^n=y^n$.  Put $x=-L/\sqrt V$ and, at the threshold $h_n(x)$ from
\eqref{eq:hn}, let
\begin{equation}
 q_n(y^n):=\Prb\{\imath_n(X^n|y^n)\ge h_n(x)\mid y^n\}.
 \label{eq:appendix-conditional-tail}
\end{equation}
This is the total probability of the sequences $x^n$ whose conditional
probabilities are at most $2^{-h_n(x)}$.  We refer to these as the
low-probability part of the conditional distribution.  Replacing $h_n(x)$
by the actual output length, or by the slightly shifted threshold used in
the achievability proof, does not change the limit.

It is useful to spell out why $f(q_n(y^n))$ appears.  Write
$p(x^n)=P_{X^n|Y^n}(x^n|y^n)$ and $c_n=2^{-h_n(x)}$.  The one-shot converse
reduces this conditional problem to approximating $p$ by a nonnegative
vector whose total mass is at most one and whose individual entries cannot
exceed $c_n$.  To see where the cap comes from, pull the ideal uniform mass
of each output bin back to the input sequences mapped into that bin.  A bin
has total reference mass $1/M_n$, so no individual sequence can receive more
than that amount.  Lemma~\ref{lem:waterfilling} says that the best vector has entries
$\min\{c_n,t_np(x^n)\}$, where $t_n\ge1$ is chosen to make the total mass
one.  Thus the largest entries of $p$ are flattened at the cap, while the
smaller entries are all multiplied by the same factor.  This is the
water-filling operation used in the proof.

Near the second-order threshold, the conditional surprisal has negligible
probability in any fixed-width interval around the water level.  Consequently,
the uncapped entries have total $p$-mass asymptotic to $q_n(y^n)$ and carry
asymptotically all of the approximating vector's mass.  Normalization then
forces $t_n$ to be close to $1/q_n(y^n)$.  On those entries, the likelihood
ratio between $p$ and the approximating vector is $1/t_n$, hence is close to
$q_n(y^n)$.  Their divergence contribution is therefore
$f(q_n(y^n))$.  The capped entries make a vanishing contribution; this is
where the assumption $f(t)/t\to0$ is used.  In the other direction,
two-universal hashing makes the same low-probability part nearly uniform and
therefore attains the same value.

For the fixed reference marginal, this argument is performed separately for
each $y^n$ and the resulting costs are then averaged with respect to
$P_Y^{\times n}$.  The limiting leakage is consequently determined by
$\E[f(q_n(Y^n))]$.  There are two sources of Gaussian fluctuation.  Given
$Y^n$, the surprisal still fluctuates around its conditional mean, with
asymptotic variance $nV_2$.  The conditional mean itself changes with
$Y^n$: with
$Z_n(Y^n)=(\mu_n(Y^n)-nH)/\sqrt n$, one has
\begin{align}
 q_n(Y^n)&\approx
 \Phi\!\left(\frac{x\sqrt V+Z_n(Y^n)}{\sqrt{V_2}}\right),
 \label{eq:nested-clt}\\
 Z_n(Y^n)&\Longrightarrow G\sim N(0,V_1).
 \label{eq:outer-clt}
\end{align}
The variance $V_2$ thus controls the Gaussian transition for a fixed
conditional distribution, while $V_1$ describes how the location of that
transition changes with the side information.  It follows that
\begin{equation}
 \E[f(q_n(Y^n))]\longrightarrow
 \E_{G\sim N(0,V_1)}
 f\!\left(\Phi\!\left(\frac{x\sqrt V+G}{\sqrt{V_2}}\right)\right)
 =\Fprof_{f,V_1/V}(x).
 \label{eq:appendix-fixed-profile}
\end{equation}
Because $f$ is applied before the average over $Y^n$, this expression can
depend separately on $V_1$ and $V_2$.

Optimizing the reference marginal changes precisely this last step.  Suppose
that the chosen reference assigns probability $Q_{Y^nS_n}(y^n,s)$ to the
public view $(y^n,s)$, whose actual probability is
$P_Y^{\times n}(y^n)P_{S_n}(s)$.  Changing the total reference mass within
this value of $(Y^n,S_n)$ multiplies every likelihood ratio by
$P_Y^{\times n}(y^n)P_{S_n}(s)/Q_{Y^nS_n}(y^n,s)$.  The same water-filling
argument therefore gives the corresponding conditional cost with this
factor, weighted by $Q_{Y^nS_n}(y^n,s)$.  Convexity of $f$ gives
\begin{equation}
 \sum_{y^n,s}Q_{Y^nS_n}(y^n,s)
 f\!\left(\frac{P_Y^{\times n}(y^n)P_{S_n}(s)q_n(y^n)}
 {Q_{Y^nS_n}(y^n,s)}\right)
 \ge f\!\left(\sum_{y^n}P_Y^{\times n}(y^n)q_n(y^n)\right).
 \label{eq:optimized-perspective-jensen}
\end{equation}
The sum on the right is $\E[q_n(Y^n)]$.  Equality is obtained by choosing
$Q_{Y^nS_n}$ as in \eqref{eq:tail-tilted-reference}.  In words, the optimized
reference gives more weight to values of $Y^n$ for which the conditional
distribution has more low-probability mass.  This makes the argument of $f$
the same for every $(y^n,s)$ and aggregates the conditional tail probabilities
before $f$ is applied.

Finally, $\E[q_n(Y^n)]$ is simply the unconditional probability that the
block surprisal exceeds the threshold.  Its variance is the total variance
$V=V_1+V_2$, so the ordinary central limit theorem gives
\begin{equation}
 f(\E[q_n(Y^n)])\longrightarrow f(\Phi(x)).
 \label{eq:appendix-optimized-profile}
\end{equation}
The optimized profile therefore depends only on $V$, even though the two
sources of fluctuation are still present in the underlying distribution.

The endpoint cases make the distinction especially transparent.  If
$V_1=0$, the conditional tail probability has no random displacement, so
averaging before or after applying $f$ gives the same value $f(\Phi(x))$.
If $V_2=0$, the conditional tail probabilities approach zero or one.  The
fixed criterion first applies $f$ to these two values and then averages,
giving $f(0)\Phi(-x)$, whereas optimization first averages the zeros and
ones and gives $f(\Phi(x))$.  The two criteria also agree for every variance
split when $f$ is affine on $[0,1]$, because in that case applying $f$ and
averaging commute; total variation is the relevant example.

\section{Near-Boundary Compatibility with Fixed-Rate Exponents}
\label{app:strong-converse-compatibility}

This appendix records a concise consistency check with the fixed-marginal
strong-converse exponent of Li, Li, and Yu \cite{li_li_yu_2025}.  This is a
comparison of asymptotic regimes, not a new uniform moderate-deviation
statement.  Consistently with the rest of the paper, let
$\alpha\in(0,1)$ be the order of the R\'enyi-divergence security criterion,
and write the auxiliary optimization parameter in their formula as
$\beta\in[\alpha,1]$.  In our notation, their two-parameter conditional
entropy is
\begin{equation}
 \widetilde H_{\beta,\alpha}(X|Y)
 :=
 \begin{dcases}
 \frac{\beta}{\alpha(1-\beta)}\log\sum_yP_Y(y)
 \left(\sum_xP_{X|Y}(x|y)^\beta\right)^{\alpha/\beta}, & \beta \in [\alpha,1),
 \\
 H(X|Y) & \beta = 1.
 \end{dcases}
 \label{eq:LLY-two-parameter-entropy}
\end{equation}
Their Theorem~19 gives, for the criterion with reference marginal $P_Y$,
\begin{equation}
 E_{\mathrm{pa}}^{(\alpha)}(R)
 =\max_{\beta\in[\alpha,1]}
 \frac{\alpha(1-\beta)}{\beta(1-\alpha)}
 \left\{R-\widetilde H_{\beta,\alpha}(X|Y)\right\}.
 \label{eq:LLY-strong-converse-exponent}
\end{equation}
Public randomization does not change this optimum: one may choose the best
deterministic hash among the seeded family.

A direct Taylor expansion of \eqref{eq:LLY-two-parameter-entropy} at
$\beta=1$, with $\beta=1-s$, gives
\begin{equation}
 \widetilde H_{1-s,\alpha}(X|Y)
 =H+\frac{\ln 2}{2}(V_2+\alpha V_1)s+O(s^2),
 \qquad s\downarrow0.
 \label{eq:LLY-entropy-near-one}
\end{equation}
By the monotonicity of $\widetilde H_{\beta,\alpha}(X|Y)$ in $\beta$
\cite[Proposition~4]{li_li_yu_2025} and \eqref{eq:LLY-entropy-near-one},
$\widetilde H_{\beta,\alpha}(X|Y)>H$ for every $\beta<1$ when $V>0$;
hence any maximizer in \eqref{eq:LLY-strong-converse-exponent} approaches
$\beta=1$ as $R\downarrow H$.
Consequently, writing $R=H+\varepsilon$ and $s=1-\beta$, the maximizing
$s$ is $O(\varepsilon)$, and maximizing the quadratic approximation in
\eqref{eq:LLY-strong-converse-exponent} yields
\begin{equation}
 E_{\mathrm{pa}}^{(\alpha)}(H+\varepsilon)
 =\frac{\alpha}{2(1-\alpha)\ln 2\,(V_2+\alpha V_1)}\,
 \varepsilon^2+o(\varepsilon^2),
 \qquad \varepsilon\downarrow0.
 \label{eq:LLY-exponent-near-boundary}
\end{equation}

The parameters $\varepsilon$ and $L$ describe the same excess output length
on different scales.  Indeed, a rate $R=H+\varepsilon$ gives an excess of
$n\varepsilon$ bits, while the second-order parametrization gives an excess
of $L\sqrt n$ bits.  Thus $\varepsilon=L/\sqrt n$, or equivalently
$n\varepsilon=L\sqrt n$.
A fixed $\varepsilon>0$ therefore corresponds to $L$ growing on the order of
$\sqrt n$.  The overlap between the two asymptotic regimes is obtained by
taking $L=L_n\to\infty$ but $L_n=o(\sqrt n)$, and hence
$\varepsilon_n=L_n/\sqrt n\downarrow0$.  Evaluating
\eqref{eq:LLY-exponent-near-boundary} on this common scale gives
\begin{equation}
 nE_{\mathrm{pa}}^{(\alpha)}
 \!\left(H+\frac{L_n}{\sqrt n}\right)
 =\frac{\alpha}{2(1-\alpha)\ln 2\,(V_2+\alpha V_1)}\,
 L_n^2+o(L_n^2).
 \label{eq:LLY-exponent-common-scale}
\end{equation}

The same expression is the far-right tail of our fixed-marginal second-order
profile.  Indeed, the Gaussian tail estimate for the profile in
\eqref{eq:G-profile}, including its two endpoint cases, is
\begin{equation}
 \log\Gprof_{\alpha,V_1/V}\!\left(-\frac{L}{\sqrt V}\right)
 =-\frac{\alpha L^2}{2\ln 2\,(V_2+\alpha V_1)}+o(L^2),
 \qquad L\to\infty.
 \label{eq:G-profile-large-L}
\end{equation}
Corollary~\ref{cor:renyi} therefore implies
\begin{equation}
 -\frac{1}{1-\alpha}
 \log\Gprof_{\alpha,V_1/V}\!\left(-\frac{L}{\sqrt V}\right)
 =\frac{\alpha}{2(1-\alpha)\ln 2\,(V_2+\alpha V_1)}\,
 L^2+o(L^2).
 \label{eq:second-order-profile-large-L}
\end{equation}
Equations~\eqref{eq:LLY-exponent-common-scale} and
\eqref{eq:second-order-profile-large-L} now have the same parameter and the
same coefficient.  This is the claimed compatibility.  It remains an
iterated-limit comparison: Li, Li, and Yu establish their operational result
at each fixed $R>H$, whereas Corollary~\ref{cor:renyi} establishes ours at
each fixed $L$.  Neither theorem alone supplies the uniformity needed to turn
the common-scale calculation into a moderate-deviation theorem, and the
fixed-rate exponent does not determine our full Gaussian profile.  Their
result uses the actual side-information marginal, so it does not address our
optimized-marginal criterion.


\begin{thebibliography}{99}
\enlargethispage{2\baselineskip}

\bibitem{csiszar_1967}
I. Csisz\'ar, ``Information-type measures of difference of probability
distributions and indirect observations,'' \emph{Studia Sci. Math. Hungar.},
vol. 2, pp. 299--318, 1967.

\bibitem{gallager_1968}
R. G. Gallager, \emph{Information Theory and Reliable Communication}.
New York, NY, USA: Wiley, 1968.

\bibitem{carter_wegman_1979}
J. L. Carter and M. N. Wegman, ``Universal classes of hash functions,''
\href{https://doi.org/10.1016/0022-0000(79)90044-8}{\emph{J. Comput. Syst.
Sci.}, vol. 18, no. 2, pp. 143--154}, Apr. 1979.

\bibitem{bennett_brassard_robert_1988}
C. H. Bennett, G. Brassard, and J.-M. Robert, ``Privacy amplification by
public discussion,''
\href{https://doi.org/10.1137/0217014}{\emph{SIAM J. Comput.}, vol. 17,
no. 2, pp. 210--229}, Apr. 1988.

\bibitem{impagliazzo_levin_luby_1989}
R. Impagliazzo, L. A. Levin, and M. Luby, ``Pseudo-random generation from
one-way functions,'' in
\href{https://doi.org/10.1145/73007.73009}{\emph{Proc. 21st Annu. ACM Symp.
Theory Comput. (STOC)}}, Seattle, WA, USA, May 1989, pp. 12--24.

\bibitem{vembu_verdu_1995}
S. Vembu and S. Verd\'u, ``Generating random bits from an arbitrary source:
Fundamental limits,''
\href{https://doi.org/10.1109/18.412679}{\emph{IEEE Trans. Inf. Theory},
vol. 41, no. 5, pp. 1322--1332}, Sep. 1995.

\bibitem{bennett_brassard_crepeau_maurer_1995}
C. H. Bennett, G. Brassard, C. Cr\'epeau, and U. M. Maurer, ``Generalized
privacy amplification,''
\href{https://doi.org/10.1109/18.476316}{\emph{IEEE Trans. Inf. Theory},
vol. 41, no. 6, pp. 1915--1923}, Nov. 1995.

\bibitem{hayashi_2008}
M. Hayashi, ``Second-order asymptotics in fixed-length source coding and
intrinsic randomness,''
\href{https://doi.org/10.1109/TIT.2008.928985}{\emph{IEEE Trans. Inf. Theory},
vol. 54, no. 10, pp. 4619--4637}, Oct. 2008.

\bibitem{watanabe_hayashi_2013}
S. Watanabe and M. Hayashi, ``Non-asymptotic analysis of privacy amplification
via R\'enyi entropy and inf-spectral entropy,'' in
\href{https://doi.org/10.1109/ISIT.2013.6620720}{\emph{Proc. IEEE Int. Symp.
Inf. Theory (ISIT)}}, Istanbul, Turkey, Jul. 2013, pp. 2715--2719.

\bibitem{tomamichel_hayashi_2013}
M. Tomamichel and M. Hayashi, ``A hierarchy of information quantities for
finite block length analysis of quantum tasks,''
\href{https://doi.org/10.1109/TIT.2013.2276628}{\emph{IEEE Trans. Inf. Theory},
vol. 59, no. 11, pp. 7693--7710}, Nov. 2013.

\bibitem{nomura_han_2013}
R. Nomura and T. S. Han, ``Second-order resolvability, intrinsic randomness,
and fixed-length source coding for mixed sources: Information spectrum
approach,''
\href{https://doi.org/10.1109/TIT.2012.2215836}{\emph{IEEE Trans. Inf. Theory},
vol. 59, no. 1, pp. 1--16}, Jan. 2013.

\bibitem{hayashi_tan_2017}
M. Hayashi and V. Y. F. Tan, ``Equivocations, exponents, and second-order coding
rates under various R\'enyi information measures,''
\href{https://doi.org/10.1109/TIT.2016.2636154}{\emph{IEEE Trans. Inf. Theory},
vol. 63, no. 2, pp. 975--1005}, Feb. 2017.

\bibitem{hayashi_tan_correction_2024}
M. Hayashi and V. Y. F. Tan, ``Corrections to `Equivocations, exponents, and
second-order coding rates under various R\'enyi\hspace{-1.2pt}\ information measures',''
\href{https://doi.org/10.1109/TIT.2024.3355506}{\emph{IEEE Trans. Inf. Theory},
vol. 70, no. 4, pp. 3046--3048}, Apr. 2024.

\bibitem{uyematsu_matsuta_2017}
T. Uyematsu and T. Matsuta, ``Second-order intrinsic randomness for correlated
non-mixed and mixed sources,''
\href{https://doi.org/10.1587/transfun.E100.A.2615}{\emph{IEICE Trans. Fundam.
Electron. Commun. Comput. Sci.}, vol. E100-A, no. 12, pp. 2615--2628}, Dec. 2017.

\bibitem{hayashi_watanabe_2016}
M. Hayashi and S. Watanabe, ``Uniform random number generation from Markov
chains: Non-asymptotic and asymptotic analyses,''
\href{https://doi.org/10.1109/TIT.2016.2530084}{\emph{IEEE Trans. Inf. Theory},
vol. 62, no. 4, pp. 1795--1822}, Apr. 2016.

\bibitem{hayashi_2016}
M. Hayashi, ``Security analysis of $\varepsilon$-almost dual universal$_2$
hash functions: Smoothing of min entropy versus smoothing of R\'enyi entropy
of order 2,''
\href{https://doi.org/10.1109/TIT.2016.2535174}{\emph{IEEE Trans. Inf. Theory},
vol. 62, no. 6, pp. 3451--3476}, Jun. 2016.

\bibitem{renes_2018}
J. M. Renes, ``On privacy amplification, lossy compression, and their duality
to channel coding,''
\href{https://doi.org/10.1109/TIT.2018.2865386}{\emph{IEEE Trans. Inf. Theory},
vol. 64, no. 12, pp. 7792--7801}, Dec. 2018.

\bibitem{yu_tan_2019}
L. Yu and V. Y. F. Tan, ``Simulation of random variables under R\'enyi
divergence measures of all orders,''
\href{https://doi.org/10.1109/TIT.2019.2891363}{\emph{IEEE Trans. Inf. Theory},
vol. 65, no. 6, pp. 3349--3383}, Jun. 2019.

\bibitem{anshu_berta_jain_tomamichel_2020}
A. Anshu, M. Berta, R. Jain, and M. Tomamichel, ``Partially smoothed
information measures,''
\href{https://doi.org/10.1109/TIT.2020.2981573}{\emph{IEEE Trans. Inf. Theory},
vol. 66, no. 8, pp. 5022--5036}, Aug. 2020.

\bibitem{abdelhadi_renes_2020}
M. Abdelhadi and J. M. Renes, ``On the second-order asymptotics of the
partially smoothed conditional min-entropy and application to quantum
compression,''
\href{https://doi.org/10.1109/JSAIT.2020.3016899}{\emph{IEEE J. Sel. Areas
Inf. Theory}, vol. 1, no. 2, pp. 416--423}, Aug. 2020.

\bibitem{nomura_2020}
R. Nomura, ``Source resolvability and intrinsic randomness: Two random number
generation problems with respect to a subclass of $f$-divergences,''
\href{https://doi.org/10.1109/TIT.2020.3009208}{\emph{IEEE Trans. Inf. Theory},
vol. 66, no. 12, pp. 7588--7601}, Dec. 2020.

\bibitem{nomura_yagi_2024}
R. Nomura and H. Yagi, ``Optimum achievable rates in two random number
generation problems with $f$-divergences using smooth R\'enyi entropy,''
\href{https://doi.org/10.3390/e26090766}{\emph{Entropy}, vol. 26, no. 9,
Art. no. 766}, Sep. 2024.

\bibitem{li_yao_2024}
K. Li and Y. Yao, ``Operational interpretation of the sandwiched R\'enyi
divergence of order $1/2$ to $1$ as strong converse exponents,''
\href{https://doi.org/10.1007/s00220-023-04890-8}{\emph{Commun. Math. Phys.},
vol. 405, no. 2, Art. no. 22}, Feb. 2024.

\bibitem{shen_gao_cheng_2022}
Y.-C. Shen, L. Gao, and H.-C. Cheng, ``Strong converse for privacy
amplification against quantum side information,'' in
\href{https://doi.org/10.1109/ISIT50566.2022.9834467}{\emph{Proc. IEEE Int.
Symp. Inf. Theory (ISIT)}}, Espoo, Finland, Jun. 2022, pp. 1880--1885.

\bibitem{shen_gao_cheng_2024}
Y.-C. Shen, L. Gao, and H.-C. Cheng, ``Optimal second-order rates for quantum
soft covering and privacy amplification,''
\href{https://doi.org/10.1109/TIT.2024.3351963}{\emph{IEEE Trans. Inf. Theory},
vol. 70, no. 7, pp. 5077--5091}, Jul. 2024.

\bibitem{li_li_yu_2025}
S.-B. Li, K. Li, and L. Yu, ``Two-parameter R\'enyi information quantities
with applications to privacy amplification and soft covering,''
\href{https://doi.org/10.1109/TIT.2026.3698938}{\emph{IEEE Trans. Inf. Theory},
vol. 72, no. 8, pp. 5370--5391}, Aug. 2026.

\bibitem{rubboli_goodarzi_tomamichel_2024}
R. Rubboli, M. M. Goodarzi, and M. Tomamichel, ``Quantum conditional
entropies from convex trace functionals,''
\href{https://doi.org/10.1007/s00220-026-05625-1}{\emph{Commun. Math. Phys.}},
vol. 407, no. 180, Aug. 2026.

\bibitem{berta_yao_2025}
M. Berta and Y. Yao, ``Strong converse exponents of partially smoothed
information measures,'' arXiv:2505.06050, 2025. [Online]. Available:
\url{https://arxiv.org/abs/2505.06050}

\bibitem{rubboli_haapasalo_tomamichel_2026}
R. Rubboli, E. Haapasalo, and M. Tomamichel, ``A complete characterisation of
conditional entropies,'' arXiv:2601.23213, 2026. [Online]. Available:
\url{https://arxiv.org/abs/2601.23213}

\bibitem{regula_tomamichel_2026}
B. Regula and M. Tomamichel, ``Rethinking quantum smooth entropies: Tight
one-shot analysis of quantum privacy amplification,'' arXiv:2603.04493, 2026.
[Online]. Available: \url{https://arxiv.org/abs/2603.04493}

\bibitem{chen_shao_2001}
L. H. Y. Chen and Q.-M. Shao, ``A non-uniform Berry--Esseen bound via Stein's
method,''
\href{https://doi.org/10.1007/PL00008782}{\emph{Probab. Theory Relat. Fields},
vol. 120, no. 2, pp. 236--254}, Jun. 2001.

\bibitem{de_moura_ullrich_2021}
L. de Moura and S. Ullrich, ``The Lean 4 theorem prover and programming
language,'' in \emph{Automated Deduction---CADE 28}, ser. Lecture Notes in
Computer Science, vol. 12699. Cham, Switzerland: Springer, 2021,
pp. 625--635.

\bibitem{mathlib_2020}
The mathlib Community, ``The Lean mathematical library,'' in \emph{Proc. 9th
ACM SIGPLAN Int. Conf. Certified Programs and Proofs (CPP)}, New Orleans, LA,
USA, Jan. 2020, pp. 367--381.

\bibitem{berta_cheng_tomamichel_lean_2026}
M. Berta, H.-C. Cheng, and M. Tomamichel, Lean formalization accompanying
``Second-Order Expansion of Privacy Amplification Under $f$-Divergence
Criteria,'' GitHub repository, archived on
\href{https://doi.org/10.5281/zenodo.22084018}{Zenodo} (2026).

\end{thebibliography}
\end{document}